\documentclass[amsmath,amssymb,onecolumn,unpublished]{quantumarticle}
\usepackage{amsthm}

\newcommand{\ket}[1]{|{#1}\rangle}
\newcommand{\bra}[1]{\langle{#1}|}
\newcommand{\ketbra}[2]{|{#1}\rangle \, \langle{#1}|}
\newcommand{\braket}[2]{\langle{#1}|{#2}\rangle}

\DeclareMathOperator{\tr}{tr} 

\newcommand{\C}{\mathcal{C}}
\newcommand{\D}{\mathcal{D}}

\newtheorem{dfn}{Definition}
\newtheorem{lemma}{Lemma}
\newtheorem{thm}{Theorem}
\newtheorem{prot}{Protocol}

\begin{document}

\title{Maximally Sensitive Sets of States}

\author{Daniel Gottesman}
\email{dgottesman@perimeterinstitute.ca}
\affiliation{Perimeter Institute, Waterloo, Canada}
\affiliation{Quantum Benchmark, Inc., Kitchener, Canada}
\affiliation{CIFAR QIS Program, Toronto, Canada}

\begin{abstract}
Coherent errors in a quantum system can, in principle, build up much more rapidly than incoherent errors, accumulating as the square of the number of qubits in the system rather than linearly.  I show that only channels dominated by a unitary rotation can display such behavior.  A \emph{maximally sensitive set of states} is a set such that if a channel is capable of quadratic error scaling, then it is present for at least one sequence of states in the set.  I show that the GHZ states in the $X$, $Y$, and $Z$ bases form a maximally sensitive set of states, allowing a straightforward test to identify coherent errors in a system.  This allows us to identify coherent errors in gates and measurements to within a constant fraction of the maximum possible sensitivity to such errors.  A related protocol with simpler circuits but less sensitivity can also be used to test for coherent errors in state preparation or if the noise in a particular circuit is accumulating coherently or not.
\end{abstract}

\maketitle

\section{Introduction}
\label{sec:intro}

As we enter into an era where experimental realizations of quantum computers are rapidly increasing in size and accuracy, the problem of characterizing errors in a quantum computer becomes more and more salient.  On the one hand, we need to know whether the device is reliable enough to allow fault-tolerant software~\cite{FT} to operate, which in turn allows logical computations to be completed with negligible errors.  The level of error also helps determine the amount of overhead (extra qubits and extra gates) needed for fault tolerance.  On the other hand, a good characterization of the errors in the system can provide guidance as to what is causing the errors and how to further reduce the error rates.

There are many different kinds of errors one might worry about.  Even in the simplest case of independent Markovian noise on each qubit, there is an infinite variety of possible types of error.  One conceptual distinction that is often made is between stochastic errors, which behave essentially probabilistically, and coherent errors, which can experience quantum interference to enhance or decrease the effect of the errors.  Then there are many errors which are a bit of both, for which the distinction between stochastic and coherent errors is not so clear.

One reason to make this distinction is because coherent errors which constructively interfere with each other can build up much faster than stochastic errors.  As a simple example, consider the $n$-qubit state $\ket{\psi_Z} = \ket{00 \ldots 0} + \ket{11 \ldots 1}$.  An example of a stochastic channel is one which does a $Z$ (phase flip) on each qubit independently with probability $p$.  Then for the state $\ket{\psi_Z}$, the probability of error when $p$ is small is roughly $pn$.  In contrast, an example coherent channel might do a phase rotation $\ket{0} \mapsto \ket{0}$, $\ket{1} \mapsto e^{i\theta} \ket{1}$.  Now the probability of error%
\footnote{Defined for this purpose to be the probability that a measurement projecting on the expected state $\ket{\psi_Z}$ gets a different result.}
for small $\theta$ is about $\theta^2 n^2$.  For a single qubit, the two channels have comparable noise rates when $p = \theta^2$, but the coherent channel rapidly outpaces the incoherent one.  Intuitively, the coherent errors allow the amplitudes of the error to add up, which are then squared via Born's rule to make probabilities, whereas for a stochastic error, the probability of error on different qubits adds up directly.  Another way of understanding the distinction is to consider the relationship between gate fidelity and diamond norm.  For some channels --- stochastic channels --- they can be similar, whereas for others --- coherent channels --- the gate fidelity can be the square of the diamond norm.

One consequence of this difference is that the threshold for fault-tolerant quantum computation to work could potentially be much worse for coherent errors than for stochastic errors~\cite{AGP}.  For instance, a fault-tolerant protocol which functions when $p < 10^{-4}$ might instead require $\theta \lesssim 10^{-4}$, which translates to a single-qubit error probability of about $10^{-8}$!  Now, we do not know for sure that this can happen in realistic situations, only that existing threshold proofs can only guarantee that the system works for the lower error rate.  It remains unclear whether it is actually possible for coherent errors in a fault-tolerant circuit to build up sufficiently close to the maximum rate to have such a big effect.  It is difficult to do simulations of coherent errors, but there has been some work on this topic (e.g., \cite{coherent,coherent2}), and those simulations suggest that the size of the effect is modest, but perhaps not negligible.  Nevertheless, it remains a possible source of concern for building large fault-tolerant quantum computers.

A number of solutions to coherent errors are known.  One possibility is dynamical decoupling~\cite{DD}, which takes advantage of the fact that coherent errors are always the same to get them to cancel out; it does not work on stochastic errors.  Another is randomized compiling~\cite{RC}, where the system is randomized periodically to prevent errors from building up coherently, effectively turning a coherent channel into a stochastic one.  Both involve extra gates or control pulses, so it makes sense to do them only if they are actually needed (both to save time and because the extra gates can potentially themselves be a source of errors).  The amount of overhead depends on the details of the gate implementations and circuits being performed; in some cases, the additional resources required by these techniques might be small or zero, but in other cases it could be more significant.

Therefore, it is useful to have ways to determine if the errors in a system are predominately coherent or mostly stochastic.  For instance, one can do a full characterization of the channel and gates in some way, for instance via gate set tomography~\cite{GST}.  However, gate set tomography is quite costly and the need to choose a gauge can cause difficulties, whereas in many applications the desired gauge is given by knowledge of the physics of the system being studied.  Another possibility is to measure the unitarity~\cite{unitarity}.  The unitarity combined with randomized benchmarking can determine the coherent error rate averaged over states and gates, but cannot determine the coherent error on a single type of gate and certainly not on state preparation or measurement.

Another common heuristic approach to detect coherent errors in a particular gate is to simply repeat the gate a number of times in sequence.  This will detect many cases of coherent errors: For instance, consistent over-rotation will accumulate coherently in this circuit.  It does not detect all coherent errors: For instance, consider a rotation about an axis which is slightly wrong.  Repeating the gate will make a larger rotation about the wrong axis, but will not amplify the error to allow it to build up rapidly even though a different sequence of gates could potentially cause a coherent build-up of this error.

A separate challenge with characterizing errors is that any serious candidate technology for building large quantum computers will have rather low error rates.  This means that a lot of data and many runs of an experiment are necessary to even distinguish the error rate from zero, and many more are needed to find out additional information about the nature of the errors.  In the case of stochastic errors, nothing can really be done about this.  When the errors  are largely coherent, however, the potential rapid rise in the error rate now becomes an advantange, allowing us to use fewer qubits and gates to identify the error.

Unfortunately, any given $n$-qubit state will not necessarily experience quadratic accumulation of errors for the particular channel in the system.  Therefore, it makes sense to ask what is a minimal set of states for which at least one will show a quadratic rise in the error rate, no matter what coherent channel it is subjected to.  In this paper, I find a set of three states, the GHZ states in the $X$, $Y$, and $Z$ bases, which have this property.  For any channel $\C$ applied to $n$ qubits, the infidelity between at least one of these three states and itself with $\mathcal{C}$ applied on all $n$ qubits is at least a constant fraction of the maximum achieveable infidelity for any state.  Therefore, by creating these three states in different sizes and looking for quadratic scaling of the error rate, we can identify the presence of coherent errors.  In fact, based on this set of measurements, we can go further and determine exactly what the coherent part of the channel is.  I will show that the coherent part is always a unitary rotation, so the characterization consists of an axis of rotation and an angle.  It turns out that this criterion is also sufficient to predict the possibility of coherent accumulation of errors in other contexts as well.

In a real system, the gates used to the create the GHZ states will be imperfect, and coherent errors in the encoding circuits could be confused with coherent errors in the channel.  The solution to this is to use randomized compiling on only \emph{part} of the circuit, the encoding part.  This converts any errors in the encoding part of the circuit into stochastic errors, leaving only coherent errors in the channel that we wish to measure.  This also enables us to measure coherent errors in gates and measurements rather than only in a straightforward near-identity channel.  To make the randomized compiling work, however, we need to assume that we have the ability to perform high-fidelity Pauli gates, with error rates low enough that they are negligible for the measurements we wish to make.  This is not an unreasonable assumption in many cases, since Pauli gates are simple single-qubit gates, but may not be always applicable in a system where all gates have a similar low error rate.  If we want to measure coherence in a fully universal set of gates, we also need an additional gate or gates such as the $\pi/4$ phase rotation to be essentially noiseless.

State preparation needs to be handled differently, since it cannot exhibit quadratic accumulation of errors with other state preparations, but can do so in conjunction with coherent errors in other gates.  This requires a formalism to discuss quadratic accumulation in more general circuits, which results in an alternative protocol for testing for coherent errors in state preparations as well as other circuit elements.  Compared to the protocol based on a maximally sensitive set of states, this alternative protocol is simpler but less sensitive.  One intriguing aspect of it is that it can be adapted to measure whether noise in a specific gate or other circuit is contributing to coherent accumulation of error in a particular circuit.  This might be helpful both theoretically for understanding the role of coherent errors in specific protocols and experimentally for identifying which coherent errors are most important to deal with for the desired application of the quantum computer.

Part of this work is quite similar to the technique and results of~\cite{LeadingKraus}, but this paper is concerned with the \emph{worst-case} behavior of channels, whereas \cite{LeadingKraus} was concerned with the behavior of channels averaged over inputs.  This is an important difference because many of the results of \cite{LeadingKraus} fail when the worst-case behavior is considered instead.  Thus it is significant that the leading Kraus approximation employed in that paper and this one continues to work when studying a channel's potential for quadratic error accumulation in the worst case, even though it does not apply to other worst-case properties of the channel.

In section~\ref{sec:quadratic}, I characterize the coherent part of a channel, and determine for any given channel the maximum amount of quadratic error scaling possible for that channel.  Then in section~\ref{sec:maxsensitive}, I show that the set of GHZ states is a maximally sensitive set of states and discuss how to implement simple protocols that are maximally sensitive to coherent errors.  Section~\ref{sec:qudits} discusses quadratic accumulation and maximally sensitive sets of states for qudits, in section~\ref{sec:different} I discuss what happens when the channel is not the same on all qubits, and in section~\ref{sec:gatenoise}, I explain how to apply randomized compiling to eliminate coherent noise in the gates we are not testing.  In section~\ref{sec:testgates}, I discuss how to use the protocol to measure coherent errors in gates and in measurements. I develop the definitions and formulas to study more general circuits in section~\ref{sec:circuits} and give a protocol for testing for coherent errors in state preparation in section~\ref{sec:contribute}.

\section{The Coherent and Incoherent Parts of a Channel}
\label{sec:quadratic}

In secs.~\ref{sec:quadratic}--\ref{sec:qudits}, I will consider the case where we have perfect quantum gates, and the only source of noise is through a quantum channel.  In particular, the protocol will consist of creating a particular state from the maximally sensitive set, followed by a quantum channel on every qubit, and then making a measurement.  Initially, I will consider a multi-qubit measurement --- in particular, a projection on the original state created --- but in section~\ref{sec:singlequbitmeasurement}, I will show how to use single-qubit measurements instead.  The gates involved in creating and measuring the states are perfect.  I will assume that the channel on each qubit is independent, and can be represented in Kraus form
\begin{equation}
\C (\rho) = \sum_k A_k \rho A_k^\dagger,
\end{equation}
with $\sum_k A_k^\dagger A_k = I$.  The Kraus operators are not unique, but some choices are better than others.  By using the Choi-Jamiolkowski isomorphism and diagonalizing the density matrix, we can choose a canonical set of Kraus operators (although even this might not be unique if the eigenvalues of the Choi matrix are degenerate).  In particular, in the limit we are interested in --- when the channel is close to the identity --- the leading Kraus operator $A_0$ is large and close to $I$ and all other Kraus operators are small.  In this case, many properties of the channel are determined by just $A_0$, allowing us to make a leading Kraus approximation~\cite{LeadingKraus}.

For a given channel $\C$, can it display quadratic error accumulation or not?  In order for this to be a well-defined question, we must work in a particular limit.  First, we must look at different values of $n$.  For low $n$ there will be some transient behavior.  For instance, the exact behavior could be of the form $an^2 + bn$.  Therefore, we want to take the limit of large $n$.  However, if the infidelity to the identity gets large, $O(1)$, it will necessarily start to saturate and the $n^2$ scaling will disappear.  For a fixed channel, this will always happen for large enough $n$.  Therefore, we need to work in a limit where $n$ is large but $n^2 r$ is small, where $r$ is the infidelity of a single use of the channel.  This means taking a family of channels $\C_r$ instead of a fixed channel.  We then look at the leading order behavior in $r$.  The coefficient of $r$ is a function of $n$, and the channel displays quadratic accumulation if it is a quadratic function.  More precisely:
\begin{dfn}
\label{def:accumulation}
Let $\C_r$ be an analytic family of quantum channels mapping a single $d$-dimensional qudit to a qudit, parametrized so that the worst-case infidelity between $\C_r$ and $I$ is $r$.  For fixed $r$ and $n$, let
\begin{equation}
F(n,r) = \min_{\ket{\psi}}  \bra{\psi} \C_r^{\otimes n} (\ketbra{\psi}{\psi}) \ket{\psi},
\end{equation}
where $\ket{\psi}$ runs over all $n$-qudit pure states.  By assumption, $F(1,r) = 1-r$.  Let $G(n) = \lim_{r \rightarrow 0}[1- F(n,r)]/r$ if the limit exists and let $G(n) = \Theta(n^a)$ for large $n$.  Then we say that $\C_r$ has \emph{order $a$ accumulation}.  In particular, if $a = 1$, $C_r$ has \emph{linear accumulation} and if $a=2$, $C_r$ has \emph{quadratic accumulation}.
\end{dfn}

While in principle this definition allows for accumulation by any polynomial, in fact, only linear and quadratic accumulation are possible.  Which one applies to a channel can be determined by looking at the leading Kraus operator $A_0$:

\begin{thm}
\label{thm:accumulation}
Let $\C_r$ be a family of quantum channels as above, and suppose that $\C_r$ has order $a$ accumulation and $d=2$ (so we are working with qubits).  Then $a=1$ or $a= 2$.  Moreover, given two such families of channels $\C_r$ and $\C'_r$ with the same leading Kraus approximation $A_0$ for all $r$, then $\C_r$ and $\C'_r$ have the same order accumulation.
\end{thm}

\begin{proof}
First, note that $a \geq 1$ always, since by considering a product state $\ket{\phi}^{\otimes n}$, we find
\begin{equation}
F(n,r) \leq \left(\min_{\ket{\phi}} \bra{\phi} \C_r (\ketbra{\phi}{\phi}) \ket{\phi}\right)^n = (1-r)^n = 1-nr + O(r^2).
\end{equation}
Thus, $G(n) \geq n$, so $a \geq 1$.

Now consider the action of $n$ copies of the full channel on the state:
\begin{equation}
\C_r^{\otimes n} (\rho) = \sum_{k_1, \ldots, k_n} \left( \bigotimes_{i=1}^n A_{k_i} \right) \rho \left( \bigotimes_{i=1}^n A_{k_i}^\dagger \right).
\end{equation}
(The dependence on $r$ of the right-hand side has been suppressed to simplify notation.)  The next step is to make a leading Kraus approximation, enabled by the following lemma, which I will prove at the end of the section.

\begin{lemma}
\label{lemma:LKA}
When $\{\C_r\}$ is an analytic family of channels such that the worst-case fidelity between $\C_r$ and $I$ is $r$, then there is a choice of Kraus operators $\{A_k\}$ for $\C_r$ such that $\|A_k\|_\infty = O(\sqrt{r})$ for $k \neq 0$.
\end{lemma}

Therefore, the only terms that can contribute to first order in $r$ are those with at most one non-zero $k_i$:
\begin{align}
\C_r^{\otimes n} (\rho) &= A_0^{\otimes n} \rho (A_0^\dagger)^{\otimes n} + 
 \sum_{i=1}^n \sum_{k\neq 0} \left( A_0^{\otimes i-1} \otimes A_k \otimes A_0^{\otimes n-i} \right) \rho \left( (A_0^\dagger)^{\otimes i-1} \otimes A_k^\dagger \otimes (A_0^\dagger)^{\otimes n-i} \right) + O(r^2) \\
&= A_0^{\otimes n} \rho (A_0^\dagger)^{\otimes n} + O(nr) + O(r^2).
\end{align}
When determining $G(n)$, the last $O(r^2)$ term disappears.  The second term $O(nr)$ contributes a term to $G$ that is at most linear in $n$.  This means that the order of accumulation is greater than $1$ iff the first term is $\Omega(n^a r)$ for $a>1$ and that the order depends only on $A_0$.

Now let us consider the effect of the first term when the input $\rho = \ketbra{\psi}{\psi}$ is a pure state:
\begin{equation}
\bra{\psi} A_0^{\otimes n} \ketbra{\psi}{\psi} (A_0^\dagger)^{\otimes n} \ket{\psi} = |\bra{\psi} A_0^{\otimes n} \ket{\psi}|^2.
\end{equation}

Using the singular value decomposition, $A_0 = W \tilde{D} V$, where $\tilde{D}$ is real and diagonal and $W$ and $V$ are unitary.  Setting $U = WV$, $D = V^\dagger \tilde{D} V$, we can rewrite this as 
\begin{align}
A_0 &= UD  \label{eq:svd} \\
U & = \cos \theta I + i \sin \theta \vec{v} \cdot \vec{\sigma} \\
D & = (1-p)I + \Delta \vec{w} \cdot \vec{\sigma}. 
\end{align}
Here, $\vec{\sigma} = (X, Y, Z)$ is the vector of Pauli matrices, 
\begin{equation}
X = \begin{pmatrix} 0 & 1 \\ 1 & 0 \end{pmatrix}, \quad Y = \begin{pmatrix} 0 & -i \\ i & 0 \end{pmatrix}, \quad Z = \begin{pmatrix} 1 & 0 \\ 0 & -1 \end{pmatrix},
\end{equation}
$\vec{v}$ and $\vec{w}$ are real norm-1 dimension 3 vectors, and $\theta$, $p$, and $\Delta$ are real.  Since $A_0^\dagger A_0 \leq I$, we have $p \geq |\Delta|$ and $p + |\Delta| \leq 1$.   Since $A_0$ is close to $I$, $\theta$, $p$, and $\Delta$ are all small.  To leading order in these quantities for each Pauli, we have
\begin{equation}
A_0 = (1 - p - \theta^2/2) I + ( i \theta \vec{v} + \Delta \vec{w}) \cdot \vec{\sigma}.
\end{equation}
Since $p \geq |\Delta|$, the term $\Delta \theta$ is automatically $o(p)$.

Let us determine the scaling of $r$ with $\theta$, $p$, and $\Delta$ for $A_0$.  Let $\ket{\psi} = \alpha \ket{0} + \beta \ket{1}$, and we choose to work in the basis where $U$ is diagonal and $D$ has no $Y$ component, so $\vec{v} = (0, 0, 1)$ and $\vec{w} = (w_X, 0, w_Z)$.  Then
\begin{equation}
\bra{\psi} A_0 \ket{\psi} = 1 - p - \theta^2/2 + (i \theta + \Delta w_Z) (|\alpha|^2 - |\beta|^2) + \Delta w_X (\alpha^* \beta + \beta^* \alpha)
\end{equation}
and
\begin{equation}
1 - |\bra{\psi} A_0 \ket{\psi}|^2 = 2p + \theta^2 - 2 \Delta w_Z (|\alpha|^2 - |\beta|^2) - 2 \Delta w_X (\alpha^* \beta + \beta^* \alpha) - \theta^2 (|\alpha|^2 - |\beta|^2)^2.
\end{equation}
In particular, we see that at least one of $p$ and $\theta^2$ must be $\Theta(r)$.  (If $\Delta$ is $\Theta(r)$, $p$ must be too since $|\Delta| \leq p$.)

Now let us look at the $n$-qubit case. To leading order,
\begin{equation}
A_0^{\otimes n} = (1-np - n\theta^2/2) I + (i \theta \vec{v} + \Delta \vec{w}) \cdot \left(\sum_{j=1}^n \vec{\sigma}_j \right) - \sum_{j<k = 1}^{n} \theta^2 (\vec{v} \cdot \vec{\sigma}_j) (\vec{v} \cdot \vec{\sigma}_k).
\label{eqn:A0n}
\end{equation}
The leading $I$ term in $|\bra{\psi} A_0^{\otimes n} \ket{\psi}|^2$ will always be $1 - \Theta(nr)$ regardless of state $\ket{\psi}$ and the relative size of $p$ and $\theta^2$.

The real part of the second term $\Delta \sum \vec{w} \cdot \vec{\sigma}_j$ contributes $O(\Delta n)$.  The remaining terms $i \theta \sum \vec{v} \cdot \vec{\sigma}_j$ and $\sum \theta^2 (\vec{v} \cdot \vec{\sigma}_j) (\vec{v} \cdot \vec{\sigma}_k)$ can contribute $O(\theta^2 n^2)$, the imaginary one giving the square of something with $n$ terms and the real one having $n^2$ terms.

If $p = \Theta(r)$ and $\theta^2 = o(r)$, then $G(n)$ is a function of $p$ and $\Delta$ and we find that $G(n) = \Theta(n)$, so $a=1$.  If $\theta^2 = \Theta(r)$ and $p = o(r)$, then $G(n)$ is a function of $\theta^2$.  In this case, $a$ could be higher than $1$, but not more than $2$.  If both $p$ and $\theta^2$ are $\Theta(r)$, then all of $p$, $\Delta$, and $\theta^2$ contribute to $G(n)$, and again $a$ could be more than $1$ but at most $2$.

In particular, when $\theta^2 = \Theta(r)$, we can look at the contribution of $i \theta \sum \vec{v} \cdot \vec{\sigma}_j$ and $\sum \theta^2 (\vec{v} \cdot \vec{\sigma}_j) (\vec{v} \cdot \vec{\sigma}_k)$ to $G(n)$.  Again in the basis where $U$ is diagonal, the GHZ state $\ket{\psi} = \frac{1}{\sqrt{2}} (\ket{0 \ldots 0} + \ket{1 \ldots 1})$ gives us $\bra{\psi} Z_j \ket{\psi} = 0$ and
\begin{equation}
\bra{\psi} \sum_{j<k} \theta^2 Z_j Z_k  \ket{\psi} = \theta^2 n(n-1)/2.
\end{equation}
Thus, $G(n) \geq n\theta^2 + n^2 \theta^2$ (allowing for the possibility of states other than the GHZ giving higher infidelity), so $a=2$ in this case.

In fact, the GHZ achieves the highest possible quadratic scaling $n^2 \theta^2$: the $i$ term can only contribute positively to $F(n,r)$ since it gets an absolute value squared.  Thus, the largest contribution it can make to infidelity is to be $0$ as with the GHZ state.  The $\sum_{jk}$ term is also maximal for the GHZ state.  For fixed $n$, it is possible that some state other than GHZ could be optimal (at least when $p$ and $\theta$ are both relevant), but as $n \rightarrow \infty$, the quadratic term will dominate and GHZ state is asymptotically optimal. 
\end{proof}

We still need to prove lemma~\ref{lemma:LKA}:
\begin{proof}[Proof of lemma]
We begin by using the Choi-Jamiolkowski isomorphism to get the state $\Phi_r = (I \otimes \C_r) (\ketbra{\phi^+}{\phi^+})$.  $\ket{\phi^+} = \frac{1}{\sqrt{d}} \sum_{j=1}^d {\ket{dd}}$ is a maximally entangled state.  Since $\C_r$ is close to the identity, $\Phi_r$ is close to $\ket{\phi^+}$, and in particular,
\begin{equation}
\bra{\phi^+} \Phi_r \ket{\phi^+} = 1 - \Theta(r).
\end{equation}
The upper and lower bounds on the fidelity are $d$-dependent, but we are considering $d$ to be a constant here.  We can diagonalize $\Phi_r$:
\begin{equation}
\Phi_r = \sum_k \lambda_k \ketbra{\phi_k}{\phi_k}.
\end{equation}
The $\ket{\phi_k}$ are orthonormal pure states and $\sum_k \lambda_k = \tr \Phi_r = 1$.
\begin{equation}
\bra{\phi^+} \Phi_r \ket{\phi^+} = \sum_k \lambda_k |\braket{\phi^+}{\phi_k}|^2
\end{equation}
Assume without loss of generality that $|\braket{\phi^+}{\phi_k}|^2$ is largest for $k=0$.  Then $\bra{\phi^+} \Phi_r \ket{\phi^+} \leq |\braket{\phi^+}{\phi_0}|^2$, so for the former to be $1-\Theta(r)$, we must have $|\braket{\phi^+}{\phi_0}|^2 = 1 - O(r)$.  Now, $\sum_k |\braket{\phi^+}{\phi_k}|^2 = 1$ (or $\leq 1$ if the $\ket{\phi_k}$ don't form a full basis), so $|\braket{\phi^+}{\phi_k}|^2 = O(r)$ for $k \neq 0$.  It must therefore also be the case that $\lambda_0 = 1 - O(r)$; but this implies that $\lambda_k = O(r)$ for $k \neq 0$.

The relevance of this is that we can choose the Kraus operators for $\C_r$ to be 
\begin{equation}
A_k \ket{\psi} = \sqrt{d \lambda_k} (\bra{\psi^*} \otimes I) \ket{\phi_k},
\end{equation}
where the $*$ means complex conjugate in some basis used to define $A_k$.  This means that $\| A_k \|_\infty = O(\sqrt{\lambda_k}) = O(\sqrt{r})$ for $k \neq 0$, proving the lemma.
\end{proof}

\section{Maximally Sensitive Set of States}
\label{sec:maxsensitive}

\begin{dfn}
\label{def:mss}
Let $S_n$ be a set of $n$-qudit states for $n = n_0, n_1, \ldots$ an infinite monotonically increasing sequence of positive integers.  Then $\{S_n\}$ is a \emph{maximally sensitive set} if, for any family of qudit channels $\C_r$ with quadratic accumulation, there exists a sequence of states $\ket{\psi_i} \in S_{n_i}$ such that $\tilde{G}(\{\ket{\psi_i}\}, n_i) = \Theta(n_i^2)$, where
\begin{align}
\tilde{G}(\{\ket{\psi_i}\}, n_i) &= \lim_{r \rightarrow 0}[1- \tilde{F}(\ket{\psi_i},n_i,r)]/r, \\
\tilde{F}(\ket{\psi_i}, n_i,r) &= \bra{\psi_i} \C_r^{\otimes n} (\ketbra{\psi_i}{\psi_i}) \ket{\psi_i}.
\end{align}
\end{dfn}
Here, $\tilde{F}$ and $\tilde{G}$ are defined just like $F$ and $G$ in def.~\ref{def:accumulation}, but for a specific sequence of states $\ket{\psi_i}$.  A maximally sensitive set of states is thus one for which, for any channel that has quadratic accumulation, there is a sequence of states drawn from the set for progressively larger numbers of qubits which witness the quadratic accumulation of error.

When $S_n$ is the full Hilbert space on $n$ qubits, it trivially forms a maximally sensitive set of states.  However, it is not immediately clear that maximally sensitive sets exist where each set $S_n$ is finite.  For instance, the proof of thm.~\ref{thm:accumulation} shows that to realize the optimal quadratic scaling one would need $S_n$ to be the set of all GHZ states in every basis, an infinite set.

If we don't demand the optimal scaling, but merely scaling within a constant factor of optimal, then for qubits, it is sufficient to take the GHZ states in just 3 bases, the eigenbases of $X$, $Y$, and $Z$:
\begin{thm}
\label{thm:MSS}
Let $S_n$ consist of the three states $\ket{\psi_{n,P}} = \frac{1}{\sqrt{2}} (\ket{0_P, 0_P, \ldots, 0_P} + \ket{1_P, 1_P, \ldots, 1_P}$, where $\ket{0_P}$ and $\ket{1_P}$ are the $+1$ and $-1$ eigenstates of the Pauli $P$, which runs over $X$, $Y$, and $Z$.  Then $\{S_n\}$ is a maximally sensitive set. Moreover, for any family of  channels $\C_r$ with quadratic accumulation, $\tilde{G}(\{\ket{\psi_{n,P}}\}, n)$ is at least $1/3$ of the optimal quadratic scaling for $\C_r$ for at least one of $P = X, Y, Z$.
\end{thm}
Here, $\ket{b_Z} = \ket{b}$, $\ket{b_X} = \frac{1}{\sqrt{2}} (\ket{0} + (-1)^b \ket{1})$, and $\ket{b_Y} =  \frac{1}{\sqrt{2}} (\ket{0} + i (-1)^b \ket{1})$.%
\footnote{Note that there is a convention choice in the relative phase of $\ket{0_Y}$ versus $\ket{1_Y}$, since we could have equally chosen $\ket{1_Y} =  \frac{1}{\sqrt{2}} (i \ket{0} + \ket{1})$, also a $-1$ eigenstate of $Y$ but one that gives us a different $\ket{\psi_{n,Y}}$.  There is a similar convention choice in $\ket{b_X}$ and $\ket{b_Z}$, but those choices are more standard.  The effects of this choice show up in sec.~\ref{sec:singlequbitmeasurement} by determining the basis needed for measurement of the state.}

\begin{proof}
We begin with eq.~\eqref{eqn:A0n}.  Then, to leading order,
\begin{equation}
\bra{\psi_{n,P}} A_0^{\otimes n} \ket{\psi_{n,P}} = (1-np-n\theta^2/2) + 0 - \binom{n}{2} \theta^2 v_P^2
\end{equation}
and
\begin{equation}
\tilde{F}(\ket{\psi_{n,P}}, n,r) = 1 - 2np- n(1-v_P^2) \theta^2 - n^2 \theta^2 v_P^2.
\end{equation}
Since $\vec{v}$ is a unit vector, $v_P^2 \geq 1/3$ for at least one $P$.

Again, the accumulation depends on which of $\theta^2$ and $p$ are $\Theta(r)$.  If only $p$ is, the accumulation is linear.  If $\theta^2$ is $\Theta(r)$, then the accumulation is quadratic, and the quadratic term in $\tilde{G}(\{\ket{\psi_{n,P}}\}, n)$ is a factor $v_P^2$ less than the value for a GHZ state in the eigenbasis of $U$.  That is, for at least one $P$, $\tilde{G}(\{\ket{\psi_{n,P}}\}, n)$ is within a factor $1/3$ of the optimal quadratic scaling.
\end{proof}

While it is not usually possible in reality to take the limit $r \rightarrow 0$, when $r$ is small, the $o(r)$ terms will be negligible.  We can expect the formulas to break down if $n^2 r$ is close to $1$ or if $p^2$ is about $\theta^2$.  Otherwise, the leading order $r$ behavior should be close to correct.

A maximally sensitive set of states with finite $S_n$ gives a natural protocol for detecting coherent errors in a family of channels, provided the channel is sufficiently close to the identity: Prepare the states in $S_n$ for a wide enough range of $n$s, apply the channel to all qubits, and look at the scaling with $n$ of the fidelity to the initial state for subsequences.  In the case of the GHZ states, look at subsequences all in the same bases; in the general case, it is sufficient to look at the subsequence consisting of the states with the highest infidelity for each $S_n$.

In the case of the three GHZ states, we can deduce additional information from the measured scaling.  The pre-factor of the quadratic scaling rate (if any) of the $P$-basis GHZ states tells us $\theta^2 v_P^2$.  Measured for all three bases, we learn $\theta^2$ and the unit vector $\vec{v}$, giving us a complete description of $U$, the coherent part of the channel.

\subsection{Measuring with Single-Qubit Measurements}
\label{sec:singlequbitmeasurement}

Even though we are currently working in an idealized model where there are no gate errors, it is inconvenient in practice to have to create multi-qubit entangled states and to make an entangled measurement.  Luckily, we can get the same information out about the accumulation order of the channel and a characterization of $U$ by making single-qubit measurements on the three GHZ states of our maximally sensitive set.

In particular, if we take a GHZ state and measure each qubit in the $X$ basis --- or equivalently perform Hadamard on each qubit and measure in the $Z$ basis --- the state in the absence of error will be a superposition of all even-weight bit strings.  Therefore, if we measure individual qubits but discard all information other than the parity of the classical output string, we have performed a rank-$2^{n-1}$ projective measurement with projectors $\Pi$, $I - \Pi$ plus additional information we have discarded.  Post-Hadamard, $\Pi = \sum_x \ketbra{x}{x}$, where the sum is over all even weight $x$.  Pre-Hadamard, this is equivalent to
\begin{equation}
\Pi_Z = \frac{1}{2} \sum_{x} (\ket{x} + \ket{\overline{x}}) (\bra{x} + \bra{\overline{x}}),
\end{equation}
where the sum is now over all $x$ and $\overline{x}$ is the bit string where all bits of $x$ have been flipped $0 \leftrightarrow 1$.  That is, the measurement distinguishes between the states in the subspace spanned by $\ket{x} + \ket{\overline{x}}$ from those in the subspace spanned by $\ket{x} - \ket{\overline{x}}$.

We can do this procedure for any of the three GHZ states, giving us three projective measurements $\Pi_P$ for the GHZ states $\ket{\psi_{n,P}}$.  For a state $\ket{\psi_{n,Z}}$, we do the $X$ basis measurement or the Hadamard and $Z$-basis measurement as described above.  For $\ket{\psi_{n,X}}$or $\ket{\psi_{n,Y}}$, we simply measure each qubit directly in the $Z$ basis.

For the basis $P$, one of the other Paulis $P'$ will take $\ket{0_P}$ to $\ket{1_P}$ and vice-versa.  For $P=Z$, $P' = X$, and for $P = X$ or $P=Y$, $P'=Z$.  The remaining Pauli $P''$ (i.e., $Y$, $X$, and $Y$ for the bases $P = X, Y, Z$ respectively) acts like $Y$ in the $Z$ basis, flipping the bit but also changing the phase.  Thus, a single $P'$ acting on $\ket{\psi_{n,P}}$ moves us within the subspace projected on by $\Pi_P$, whereas a single $P$ or $P''$ takes us out.   Let $\ket{\psi_{n,P}(j)} = P'_j \ket{\psi_{n,P}} = \frac{1}{\sqrt{2}} (\ket{e_{j,P}} + \ket{\overline{e}_{j,P}})$, where $e_j$ is the bit string with a $1$ in the $j$th bit and $0$s elsewhere.  Similarly, let $\ket{\psi_{n,P}(j,k)} = P'_j P'_k \ket{\psi_{n,P}}$.

Again starting from  eq.~\eqref{eqn:A0n}, we find
\begin{align}
\Pi_P A_0^{\otimes n} \ket{\psi_{n,P}} &= (1-np - n\theta^2/2)  \ket{\psi_{n,P}} + (i\theta v_{P'} + \Delta w_{P'}) \sum_j \ket{\psi_{n,P}(j)} - \\
& -\sum_{j < k} \theta^2 \left[v_P^2 \ket{\psi_{n,P}} + v_P v_{P''} (\ket{\psi_{n,P}(j)} + \ket{\psi_{n,P}(k)}) + (v_{P'}^2 + v_{P''}^2) \ket{\psi_{n,P}(j,k)} \right]. \nonumber
\end{align}
The probability squared of this outcome is therefore
\begin{equation}
(1-np - n\theta^2/2 - \theta^2 \binom{n}{2} v_P^2)^2 + n \theta^2 v_{P'}^2 + n [\Delta w_{P'} - \theta^2 (n-1) v_P v_{P''}]^2 + \binom{n}{2} \theta^4 (v_{P'}^2 + v_{P''}^2)^2.
\end{equation}
Discarding the lower-order terms, we find the probability squared of the $\Pi_P$ outcome is
\begin{equation}
1 - 2np - n\theta^2 (1-v_{P}^2 - v_{P'}^2) - n^2 \theta^2 v_P^2.
\end{equation}
Once again, the coefficient of the $n^2$ term tells us $\theta^2 v_P^2$, and running over $P = X,Y,Z$, we can deduce $\theta$ and $\vec{v}$.

Using this technique, we can identify coherent errors as follows:
\begin{prot}[Identifying coherent errors in a channel]
\label{prot:channel}
Choose a sequence of values for $n$.  They should span a large enough range to distinguish quadratic from linear behavior and the largest $n$ should satisfy $n^2 r \ll 1$.

\begin{enumerate}
\item For each value of $n$, create copies of $\ket{\psi_{n,X}}$, $\ket{\psi_{n,Y}}$, and $\ket{\psi_{n,Z}}$.
\item Run all $n$ qubits of each state through the channel $\C_r$.
\item Measure the qubits of $\ket{\psi_{n,X}}$ and $\ket{\psi_{n,Y}}$ in the $Z$ basis and measure the qubits of $\ket{\psi_{n,Z}}$ in the $X$ basis.
\item Count the outcome as an error if the outcome bit string for a particular state has odd parity.  Determine the error rate for each state.
\item Fit the error rates for different $n$ for each basis $P$ to a quadratic $a_P n^2 + b_P n + c_P$.
\item Let $\theta^2 = a_X + a_Y + a_Z$ and $v_P^2 = a_P / \theta^2$.  These are the parameters of the coherent error in the channel $\C_r$. 
\end{enumerate}
\end{prot}

\subsection{Using Fewer Qubits}
\label{sec:fewerqubits}

Another drawback of the protocol as presented so far is that it requires many qubits.  If we have the ability to measure and reset qubits during the computation, this can be substantially reduced and the protocol can be easily run using just two qubits.  This is because the GHZ states in our standard maximally sensitive set can be created via a straightforward circuit as in fig.~\ref{fig:GHZfull} which starts one qubit in the state $\ket{0} + \ket{1}$ and then does a sequence of CNOTs to each of the other $n-1$ qubits, which all start in the state $\ket{0}$.  Each qubit can then be rotated into the basis for the appropriate GHZ state.  Then the channel is applied, and finally, we rotate each individual qubit to the appropriate basis for measurement.

That means that except for the first qubit, each qubit experiences exactly one CNOT gate, followed by a sequence of single-qubit gates and measurement.  Therefore, we can perform the CNOT to the $i$th qubit, follow it by a single-qubit rotation to get the $i$th qubit into the right basis if necessary, apply the channel, rotate again to the measurement basis, and then measure the $i$th qubit, all without needing to do anything to the $j$th qubit for $j>i$.  That is, we can complete the measurement for the $i$th qubit before even starting to interact the $(i+1)$th qubit with the first qubit.  The procedures for qubits $2$ through $n$ can be done sequentially on a single physical qubit, recording each measurement and then resetting the physical qubit to $\ket{0}$, as in fig.~\ref{fig:GHZshrunk}.  The first qubit must maintain its coherence throughout this procedure.  At the end, once qubits $2$ through $n$ have all been entangled and measured, then we do the single-qubit operations for the first qubit: rotate to the appropriate basis for the GHZ state, apply the channel to it, rotate it to the measurement basis, and measure.  We calculate the parity of all $n$ measurement outcomes, the $n-1$ measurements on the second physical qubit as well as the final measurement on the first qubit.

While this version of the procedure is a bit more difficult to apply in a communications scenario where we are interested in channel noise, it is a very useful variant when we are interested in testing for coherent errors in gates, as in sec.~\ref{sec:testgates}.

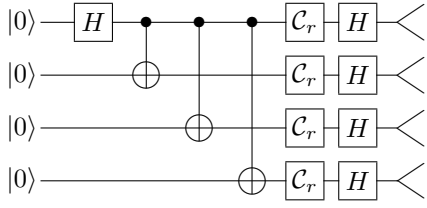
\begin{figure}
\begin{picture}(180,100)
\put(20,20){\line(1,0){93}}
\put(20,40){\line(1,0){93}}
\put(20,60){\line(1,0){93}}
\put(20,80){\line(1,0){13}}

\put(6,73){\makebox(14,14){$\ket{0}$}}
\put(6,53){\makebox(14,14){$\ket{0}$}}
\put(6,33){\makebox(14,14){$\ket{0}$}}
\put(6,13){\makebox(14,14){$\ket{0}$}}

\put(33,73){\framebox(14,14){$H$}}
\put(47,80){\line(1,0){66}}

\put(60,80){\circle*{4}}
\put(60,80){\line(0,-1){25}}
\put(60,60){\circle{10}}

\put(80,80){\circle*{4}}
\put(80,80){\line(0,-1){45}}
\put(80,40){\circle{10}}

\put(100,80){\circle*{4}}
\put(100,80){\line(0,-1){65}}
\put(100,20){\circle{10}}

\put(113,13){\framebox(14,14){$\C_r$}}
\put(113,33){\framebox(14,14){$\C_r$}}
\put(113,53){\framebox(14,14){$\C_r$}}
\put(113,73){\framebox(14,14){$\C_r$}}

\put(127,20){\line(1,0){6}}
\put(127,40){\line(1,0){6}}
\put(127,60){\line(1,0){6}}
\put(127,80){\line(1,0){6}}

\put(133,13){\framebox(14,14){$H$}}
\put(133,33){\framebox(14,14){$H$}}
\put(133,53){\framebox(14,14){$H$}}
\put(133,73){\framebox(14,14){$H$}}

\put(147,20){\line(1,0){8}}
\put(147,40){\line(1,0){8}}
\put(147,60){\line(1,0){8}}
\put(147,80){\line(1,0){8}}

\put(155,80){\line(3,2){10}}
\put(155,80){\line(3,-2){10}}
\put(155,60){\line(3,2){10}}
\put(155,60){\line(3,-2){10}}
\put(155,40){\line(3,2){10}}
\put(155,40){\line(3,-2){10}}
\put(155,20){\line(3,2){10}}
\put(155,20){\line(3,-2){10}}

\end{picture}
\caption{The protocol for the state $\ket{\psi_{4,Z}}$ using $4$ physical qubits}
\label{fig:GHZfull}
\end{figure}

\begin{figure}
\begin{picture}(350, 60)
\put(20,40){\line(1,0){13}}
\put(20,20){\line(1,0){53}}

\put(6,33){\makebox(14,14){$\ket{0}$}}
\put(6,13){\makebox(14,14){$\ket{0}$}}

\put(33,33){\framebox(14,14){$H$}}
\put(47,40){\line(1,0){246}}

\put(60,40){\circle*{4}}
\put(60,40){\line(0,-1){25}}
\put(60,20){\circle{10}}

\put(73,13){\framebox(14,14){$\C_r$}}
\put(87,20){\line(1,0){6}}
\put(93,13){\framebox(14,14){$H$}}
\put(107,20){\line(1,0){8}}

\put(115,20){\line(3,2){10}}
\put(115,20){\line(3,-2){10}}

\put(155,20){\line(1,0){28}}
\put(141,13){\makebox(14,14){$\ket{0}$}}

\put(170,40){\circle*{4}}
\put(170,40){\line(0,-1){25}}
\put(170,20){\circle{10}}

\put(183,13){\framebox(14,14){$\C_r$}}
\put(197,20){\line(1,0){6}}
\put(203,13){\framebox(14,14){$H$}}
\put(217,20){\line(1,0){8}}

\put(225,20){\line(3,2){10}}
\put(225,20){\line(3,-2){10}}

\put(265,20){\line(1,0){28}}
\put(251,13){\makebox(14,14){$\ket{0}$}}

\put(280,40){\circle*{4}}
\put(280,40){\line(0,-1){25}}
\put(280,20){\circle{10}}

\put(293,13){\framebox(14,14){$\C_r$}}
\put(307,20){\line(1,0){6}}
\put(313,13){\framebox(14,14){$H$}}
\put(327,20){\line(1,0){8}}

\put(335,20){\line(3,2){10}}
\put(335,20){\line(3,-2){10}}

\put(293,33){\framebox(14,14){$\C_r$}}
\put(307,40){\line(1,0){6}}
\put(313,33){\framebox(14,14){$H$}}
\put(327,40){\line(1,0){8}}

\put(335,40){\line(3,2){10}}
\put(335,40){\line(3,-2){10}}

\end{picture}
\caption{The protocol for the state $\ket{\psi_{4,Z}}$ using $2$ physical qubits}
\label{fig:GHZshrunk}
\end{figure}
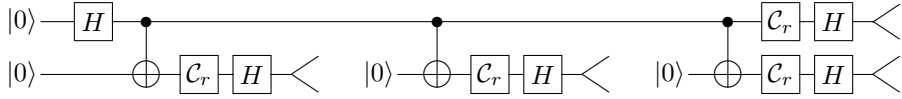

\section{Qudits}
\label{sec:qudits}

Defs.~\ref{def:accumulation} and \ref{def:mss} are phrased for qudits of arbitrary dimension $d$, but I have so far proven thms.~\ref{thm:accumulation} and \ref{thm:MSS} only for qubits.  The main complicating factor in going to qudits with $d > 2$ is that the description of $U$ and $D$ appearing through the singular value decomposition of $A_0$, eq.~\eqref{eq:svd}, is more complicated.

Before getting started, we pick some set of operators $\mathcal{P}$ to generalize the Paulis.  I will require that $P \in \mathcal{P}$ is unitary, that the operators in $\mathcal{P}$ form a basis for the space of $d \times d$ matrices, that $I \in \mathcal{P}$, and that $P \in \mathcal{P}$ is traceless if $P \neq I$.  One possible choice is the Heisenberg-Weyl operators $P = X^a Z^b$. $a, b = 0, \ldots d-1$,
\begin{align}
X \ket{j} &= \ket{j+1 \bmod d} \\
Z \ket{j} &= \omega^j \ket{j} \\
\omega &= e^{2\pi i/d}.
\end{align}
Another option if $d = 2^m$ is to treat the qudit as $m$ qubits and use the $m$-qubit Pauli group.  (This can be generalized to powers of odd primes as well.)  Since the $P$s form a basis for operators on the $d$-dimensional Hilbert space, for any $M$, we can write $M = \sum_P m_P P$.  If $M$ is Hermitian, $M = M^\dagger$, which implies that $m_P = m_{P^\dagger}^*$.  

Now returning to determining error accumulation, $D$ is Hermitian and $U$ is unitary, so we can write $U = \exp(i \theta H)$, where $H$ is Hermitian as well.  We let $H = \sum_P h_P P$ and $D = \sum_P d_P P$.  We can choose $\theta$ so that $\sum_P |h_P|^2 = 1$.  We can assume $h_I = 0$ since that merely contributes a global phase to $U$, which can be removed from $A_0$.  Also, let $d_I = 1-p$.   Then in order to have $A_0^\dagger A_0 \leq I$, it must be the case that all $d_P = O(p)$.   

Again expanding to leading order in each Pauli, we find
\begin{equation}
A_0 = (1 - p - \frac{\theta^2}{2}) I + \sum_{P \neq I} (i \theta h_P + d_P) P
\label{eq:A0qudit}
\end{equation}
and
\begin{equation}
A_0^{\otimes n} = (1 - np - \frac{n\theta^2}{2}) I + \sum_{P \neq I} \sum_{j} (i \theta h_P + d_P) P_j - \sum_{P, Q \neq I} \sum_{j<k} \theta^2 h_P h_Q P_j Q_k.
\label{eq:A0nqudit}
\end{equation}

For a single qudit, let $\ketbra{\psi}{\psi} = \sum_P c_P P$.  Since $P$ is traceless unless $P=I$, $c_I = 1/d$.  We have (discarding lower-order terms)
\begin{align}
\tr A_0 \ketbra{\psi}{\psi} &= 1 - p - \frac{\theta^2}{2} + d \sum_{P \neq I} (i \theta h_P + d_P) c_{P^\dagger} \\
\big|\tr A_0 \ketbra{\psi}{\psi} \big|^2 &= 1 - 2p - \theta^2 + d \sum_{P \neq I} [i \theta (h_P c_{P^\dagger} - h_{P^\dagger} c_P) + d_P c_{P^\dagger} + d_{P^\dagger} c_P] + \nonumber \\
& \qquad + d^2 \theta^2 \sum_{P,Q \neq I} h_P h_{Q^\dagger} c_{P^\dagger} c_Q \\
&= 1 - 2p - \theta^2  + 2 d \sum_{P \neq I} d_P c_{P^\dagger} + d^2 \theta^2 \sum_{P,Q \neq I} h_P h_{Q^\dagger} c_{P^\dagger} c_Q,
\label{eq:A0worstcasequdit}
\end{align}
since $\sum h_P c_{P^\dagger} = \sum h_{P^\dagger} c_P$ and $\sum d_P c_{P^\dagger} = \sum d_{P^\dagger} c_P$ by reordering the sums.  Thus, as for qubit channels, we must have at least one of $p$ and $\theta^2$ be $\Theta(r)$.

Let $P$ and $Q$ be non-commuting Heisenberg-Weyl operators which generate the Heisenberg-Weyl group, so any Heisenberg-Weyl operator can be written as $P^a Q^b$ up to a phase.  Let $\ket{\psi_{n,P,Q}}$ be the $n$-qudit stabilizer state with stabilizer generated by $P_j P^\dagger_{j+1}$ (for $j=1, \ldots, n-1$) and $Q^{\otimes n}$.  Then
\begin{align}
\big|\tr A_0^{\otimes n} \ketbra{\psi_{n,P,Q}}{\psi_{n,P,Q}} \big|^2 = 1 - 2np - n\theta^2  - 2 \binom{n}{2} \theta^2 \sum_{s=1}^{d-1} h_{P^s} h_{P^{-s}}.
\end{align}
Let $v_P = \sum_{s=1}^{d-1}  h_{P^s} h_{P^{-s}} = \sum |h_{P^s}|^2 \geq 0$.  We now wish to find a set of $S$ of $P$s such that there is always guaranteed to be a $v_P > 0$ for some $P$.

 Let $S$ be a set of Heisenberg-Weyl operators such that $\forall$ Heisenberg-Weyl operators $R$, there is an integer $s$ and a Heisenberg-Weyl operator $P \in S$ such that $R = P^s$.  For instance, for any $d$ we can take $S =\{ X^a Z^b\}$ for all $a, b$ provided both are not $0$. We have $\sum_{P \in S} v_P \geq \sum_{P \neq I} |h_P|^2 = 1$.  Thus, for some $P$, $v_P \geq 1/(d^2-1)$.   This shows that if $\theta^2$ is $\Theta(r)$, the family of channels has quadratic accumulation, and there will be at least one state $\ket{\psi_{n,P,Q}}$ with $P \in S$ that exhibits the quadratic accumulation.

Smaller sets $S$ always exist.  In particular, if $d$ is prime, $S = \{ X, Z, X^a Z\}$ for $a= 1, \ldots, d-1$ works.  In this case, $\sum_{P \in S} v_P = \sum_{P \neq I} |h_P|^2 = 1$.  The quadratic scaling of $\tilde{G}(\{\ket{\psi_{n,P,Q}}\}, n)$ is $v_P \theta^2$.  There must be at least one $v_P \geq 1/(d+1)$, so we get a tighter bound on the quadratic scaling.

For $d$ a power of $2$, we get a similar result using a different set of GHZ-like states.  The same can also be done for powers of other primes.  Let $B$ run over a complete set $S$ of mutually unbiased bases (MUBs)~\cite{MUB} for $m$ qubits, each of which is a stabilizer basis.  Let $\{M_a\}$, $a = 1, \ldots, m$ be a set of generators of the stabilizer for one state of the basis $B$ and let $Q_a$ be an $m$-qubit Pauli that anticommutes with $M_a$ and commutes with $M_b$ and $Q_b$ for $b \neq a$.  A set of such $Q_a$ always exists and can be found by solving the set of linear equations derived from these anticommutation constraints.  For convenience, we let $B$ also denote the abelian group generated by $\{M_a\}$, the stabilizer of the chosen state from the basis $B$; other states in $B$ have the same stabilizer with different eigenvalues.  Note that any Pauli on a single qudit will anticommute with at least one of the $M_a$ or $Q_a$.

Let $M_{a,j}$ be $M_a$ acting on the $j$th $d$-dimensional qudit.  Then let $\ket{\psi_{n,B}}$  be the $n$-qudit stabilizer state with generators $M_{a,j} M_{a,j+1}$ and $Q_a^{\otimes n}$ for $a=1, \ldots, m$ and $j = 1, \ldots, n-1$.  This state is equivalent to $m$ standard $n$-qubit GHZ states under a local transformation on each $d$-dimensional qudit.

Note that $\ket{\psi_{n,B}}$ may depend somewhat on the choice of $Q_a$, which are not completely unique.  If $Q_a$ anticommutes with $M_a$ and commutes with $M_b$, so does $Q_a M$ whenever $M$ is an element of $B$, and all Paulis that have the same commutation relations with the $M_a$s have this form.  If $n$ is even, all choices of $Q_a$ give the same stabilizer state $\ket{\psi_{n,B}}$.  If $n$ is odd, however, each choice of $\{Q_a\}$ gives a different state.  The choice of which state from the basis $B$ determines $\{M_a\}$ does not matter, nor does the choice of basis for $B$.

This time, we find that 
\begin{align}
\big|\tr A_0^{\otimes n} \ketbra{\psi_{n,P,Q}}{\psi_{n,P,Q}} \big|^2 = 1 - 2np - n\theta^2  - 2 \binom{n}{2} \theta^2 \sum_{M \in B} h_{M}^2.
\end{align}
Since $M$ is Hermitian in this case, $h_M$ is real.  Let $v_B = \sum h_M^2$.  For a complete set of stabilizer mutual unbiased bases, every Pauli $M$ is in some $B \in S$.  As $B$ runs over $S$, therefore, we get all non-trivial Paulis and $\sum_B v_B = 1$.  A complete set of mutually unbiased bases contains $d+1$ bases, so once again, at least one $v_P \geq 1/(d+1)$.

In short, we have shown the following theorem generalizing  thms.~\ref{thm:accumulation} and \ref{thm:MSS}:

\begin{thm}
\label{thm:quditaccumulation}
Let $\C_r$ be a family of quantum channels as before for qudits of dimension $d$, and suppose that $\C_r$ has order $a$ accumulation.  Then $a=1$ or $a= 2$.  Moreover, given two such families of channels $\C_r$ and $\C'_r$ with the same leading Kraus approximation $A_0$ for all $r$, then $\C_r$ and $\C'_r$ have the same order accumulation.  That is, the order depends only on the leading Kraus operator. 

Let $S_n = \{\ket{\psi_{n,P,Q}}\}$ with $P \in S$ and $Q$ any Heisenberg-Weyl operator that generates the Heisenberg-Weyl group with $P$; or when $d = 2^m$, let $S_n = \{\ket{\psi_{n,B}}\}$ with $B \in S$.  Then $\{S_n\}$ is a maximally sensitive set.  For at least one $P \in S$, the quadratic scaling of $\tilde{G}(\{\ket{\psi_{n,P,Q}}\}, n)$ or $\tilde{G}(\{\ket{\psi_{n,B}}\}, n)$ is at least $1/(d^2-1)$ of the maximum achievable.  For $d$ prime or a power of two, we can tighten this to at least $1/(d+1)$ of the maximum achievable for any state.
\end{thm}

Once again, by measuring the quadratic scaling for all $d+1$ states, we can recover $\theta^2$ through the sum of the scalings.  We cannot fully reconstruct $U$, however, since we only learn the values $v_P$ or $v_B$.  It is likely that by adding more states to make a larger maximally sensitive set, we could gain full information about $U$.  David Poulin has pointed out~\cite{Poulin} that for the case $d= 2^m$, we could treat the channel as instead a single-qubit channel by tracing out $m-1$ of the qubits composing the qudit.  By applying the single-qubit protocol to different decompositions of the qudit into $m$ qubits, this should enable us to characterize the coherent error in the channel completely.

As with the qubit case, we can make single-qudit measurements to learn the scaling instead of entangled measurements.  If the state used is $\ket{\psi_{n,P,Q}}$, the accepting projector is created by measuring in the $Q$ basis and counting all results for which the dits of the output string add up to a multiple of $d$.  When $d$ is prime, this leads to the projector
\begin{equation}
\Pi_Q = \frac{1}{d} \sum_{s=0}^{d-1} (Q^s)^{\otimes n}.
\end{equation}
We calculate $\bra{\psi_{n,P,Q}} (A_0^\dagger)^{\otimes n} \Pi_Q A_0^{\otimes n} \ket{\psi_{n,P,Q}}$ and once again find that the leading order coefficient of $n^2$ is $\theta^2 v_P$.  For a state $\ket{\psi_{n,B}}$, we need to measure the eigenvalue of $Q_a$ for each qudit and then check that the number of $+1$ eigenvalues of $Q_a$ is even.  We must do this for all $a$, and this cannot generally be done with single-qubit measurements.  However, if we rotate each qudit so that $Q_a = Z_a$ for the $a$th qubit in the qudit, then a standard basis measurement of each qubit suffices.  Again, we find that the leading order coefficient of $n^2$ in the infidelity is $\theta^2 v_B$.

\section{Generalizing the Error Model}

In this section, I will consider more general error models than the original idealized one of perfect gates with the same noisy channel on each qubit.

\subsection{Different Channels for Different Qubits}
\label{sec:different}

In a real system, the channel will not be identical for all qubits.  If each qubit $j$ has a different $p_j$, $\Delta_j$, $\theta_j$, $\vec{v}_j$, and $\vec{w}_j$, then eq.~\eqref{eqn:A0n} is replaced by
\begin{equation}
A_0^{\otimes n} = 1 - \sum_j (p_j - \theta_j^2/2) I + \sum_j (i \theta_j \vec{v}_j \cdot \vec{\sigma}_j  + \Delta_j \vec{w}_j \cdot \vec{\sigma}_j ) - \sum_{j<k} \theta_j \theta_k (\vec{v}_j \cdot \vec{\sigma}_j) (\vec{v}_k \cdot \vec{\sigma}_k).
\end{equation}
For the three GHZ states $\ket{\psi_{n,P}}$, again only the last sum contributes to the quadratic scaling, giving
\begin{align}
\tilde{F}(\ket{\psi_{n,P}}, n,r) &= 1 - 2\sum_{j<k} \theta_j \theta_k v_{P,j} v_{P,k} + o(n^2 \theta^2) \\
&= 1 - \left(\sum_j \theta_j v_{P,j} \right)^2 + \sum_j \theta_j^2 v_{P,j}^2 + o(n^2 \theta^2).
\end{align}
The sum over $\theta_j^2 v_{P,j}^2$ only has $n$ terms and therefore cannot contribute to quadratic scaling.  If the $v_{P,j}$ vary wildly, and in particular if they often differ in sign, the sum $\sum \theta_j v_{P,j}$ will largely cancel and there will be no quadratic scaling at all.  If the $v_{P,j}$ and $\theta_j$ all have the same sign and do not generally decrease in size for large $n$, then there will be some quadratic scaling, with a coefficient given by the square of the average of $\theta_j v_{P,j}$.

Let $\theta = \frac{1}{n} \sum_j \theta_j$, $v_P = \frac{1}{n} \sum_j v_{P,j}$, $\theta_j = \theta + \delta_j$, $v_{P,j} = v_P + \eta_{P,j}$.  Then
\begin{equation}
\sum_j \theta_j v_{P,j} = n \theta v_P + \sum_j \delta_j \eta_{P,j},
\end{equation}
since $\sum \delta_j = \sum \eta_{P,j} = 0$.  When we measure the scaling for all 3 GHZ states, we learn this quantity for $P = X, Y, Z$.

Now, $\sum_P v_{P,j}^2 = 1$, so
\begin{equation}
n = \sum_{P,j} v_{P,j}^2 = n \sum_P v_P^2 + \sum_{P,j} \eta_{P,j}^2.
\end{equation}
This means that $\vec{v} = (v_X, v_Y, v_Z)$ need not be a unit vector,
\begin{equation}
|\vec{v}|^2 = 1 - \frac{1}{n} \sum_{P,j} \eta_{P,j}^2,
\end{equation}
but if all $\eta_{P,j}$ are small, it is close to a unit vector, up to the variance squared.  When $\delta_j$ is likewise small, we thus learn the average $\theta$ and $v_P$ up to second order in the perturbations.

\subsection{Noisy Gates}
\label{sec:gatenoise}

If the gates used to create the states are imperfect, there is the possibility that we could see a quadratic accumulation of error due to noise in the gates rather than the channel being tested.  In the most general case, if we use at least $n^2$ gates to create or measure the states in our maximally sensitive set, then any kind of error in the gates could lead to a $\Theta (n^2)$ scaling in the error rate measured, even without any coherent part to the noise.  This seems unavoidable if we insist on using states which require this many gates to create or measure.   

Luckily, the GHZ states from sec.~\ref{sec:maxsensitive} can be created using $O(n)$ gates, so this particular issue does not affect the protocol derived from them.  However, it is still the case that coherent noise in the gates used to create or measure the GHZ states could be confused with coherent noise in the channel.  In this section, I will discuss how to remove the ambiguity.

In particular, I will work in a model for which Pauli gates are essentially noiseless, meaning they have noise $o(r)$, so negligible in the limit $r \rightarrow 0$.  All other gates, state preparations, and measurements may have arbitrary noise that affects only the qubits involved in the gate.  It is actually sufficient for the noise on the Paulis to be gate-independent, so all Paulis experience the same noise process; in this case, the noise on the Paulis can be absorbed into noise on the other gates~\cite{RC}.

We model noisy gates as a unitary $G$ followed by a quantum channel $\D$.  $G$ is the ideal action of the gate in the absence of noise.  $\D$ is a channel that acts on the qubits affected by the gate, so in the case of a $2$-qubit gate, it is an arbitrary $2$-qubit channel.  In order to put this in the same framework as the rest of the discussion in this paper, we in fact imagine a family of noisy gates with the ideal gate $G$ followed by a quantum channel $\D_s$, where $s$ is now the worst-case fidelity of the gate.  If $s$ and $r$ scale at different rates, the whole discussion is moot --- if $s = o(r)$, the gate noise does not contribute at all to the leading order behavior of the system with $r$, whereas if $r = o(s)$, we have little hope of seeing the noise in the channel over the noise in the gates.  Therefore, we assume $s = \Theta(r)$.  For simplicity, I will assume that all gates have the same noise $\D_s$, but this is not necessary except for notational convenience.

We can analyze $\D_s$ in a similar way to how we analyzed $\C_r$ in thm.~\ref{thm:accumulation} and sec.~\ref{sec:qudits}.  We look at the first Kraus operator $B_0$ of $\D_s$ and decompose into a coherent unitary part $\exp(i \omega K)$ and a Hermitian part.  A similar analysis shows that, when there are $O(n)$ gates, there is no possibility that $\D_s$ can make a quadratic contribution to the error rate measured in the system unless $\omega^2 = \Theta(r)$.  However, when $\D_s$ does have a coherent part of the requisite size, there is at least the possibility that it will make a quadratic contribution and thus be potentially confused with coherent error in the channel.  Whether or not this actually does happen will depend on the interaction of the noise with the circuit being performed.

In order to prevent this from happening, we will modify the circuit in order to avoid any quadratic accumulation of noise from the gates.  Of course, we must be careful to do so in a way that still allows errors from the channel to accumulate quadratically.

The approach will be based on randomized compiling~\cite{RC}.  We take advantage of the fact that the only gates we need to encode or measure a GHZ state (in any of the three bases) are Clifford group gates, which have the property that they conjugate Pauli operators into other Pauli operators.  In particular, imagine we could perform an independent uniformly random two-qubit Pauli operator $P$ before and after each gate noise step.  Instead of noise $\D_s$, we would have $\D'_s = 1/16 \sum_P P \D_s P$.  $\D'_s$ is a stochastic channel with first Kraus operator $B_0 = q I$.  (All other Kraus operators are non-trivial Paulis.)  It will never cause quadratic accumulation of errors from $O(n)$ gates.  In the limit of $s \rightarrow 0$, $\D'_s \rightarrow I$, so the Paulis have no effect on the gate when there is no error.

Of course, we cannot insert the Pauli $P$ between the gate $G$ and the noise, since the separation between them is just a mathematical fiction.  However, $P G = G Q$, where $Q = G^\dagger P G$ is another Pauli operator (on two qubits if $G$ is a two-qubit gate).  Therefore, if we perform $Q$ before the gate and $P$ after the gate and noise, we have converted the noise in the gate into stochastic noise.  If we do this for every noisy gate (using independently chosen $P$ and $Q$), all the gate noise becomes of a form that cannot interfere with the quadratic accumulation signal we are trying to see.  We can simplify the resulting circuits slightly by combining sequential Paulis.

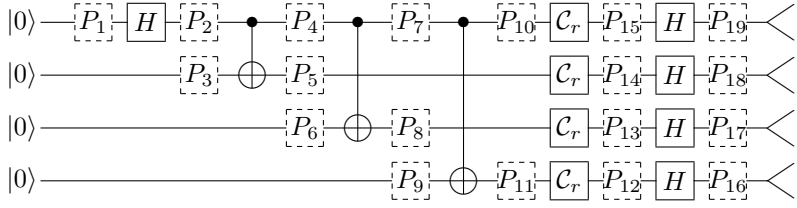
\begin{figure}
\begin{picture}(320,100)
\put(20,20){\line(1,0){133}}
\put(20,40){\line(1,0){93}}
\put(20,60){\line(1,0){53}}
\put(20,80){\line(1,0){13}}

\put(6,73){\makebox(14,14){$\ket{0}$}}
\put(6,53){\makebox(14,14){$\ket{0}$}}
\put(6,33){\makebox(14,14){$\ket{0}$}}
\put(6,13){\makebox(14,14){$\ket{0}$}}

\put(33,73){\dashbox{2}(14,14){$P_1$}}
\put(47,80){\line(1,0){6}}

\put(53,73){\framebox(14,14){$H$}}
\put(67,80){\line(1,0){6}}

\put(73,73){\dashbox{2}(14,14){$P_2$}}
\put(73,53){\dashbox{2}(14,14){$P_3$}}
\put(87,80){\line(1,0){26}}
\put(87,60){\line(1,0){26}}

\put(100,80){\circle*{4}}
\put(100,80){\line(0,-1){25}}
\put(100,60){\circle{10}}

\put(113,73){\dashbox{2}(14,14){$P_4$}}
\put(113,53){\dashbox{2}(14,14){$P_5$}}
\put(113,33){\dashbox{2}(14,14){$P_6$}}
\put(127,80){\line(1,0){26}}
\put(127,60){\line(1,0){86}}
\put(127,40){\line(1,0){26}}

\put(140,80){\circle*{4}}
\put(140,80){\line(0,-1){45}}
\put(140,40){\circle{10}}

\put(153,73){\dashbox{2}(14,14){$P_7$}}
\put(153,33){\dashbox{2}(14,14){$P_8$}}
\put(153,13){\dashbox{2}(14,14){$P_9$}}
\put(167,80){\line(1,0){26}}
\put(167,40){\line(1,0){46}}
\put(167,20){\line(1,0){26}}

\put(180,80){\circle*{4}}
\put(180,80){\line(0,-1){65}}
\put(180,20){\circle{10}}

\put(193,73){\dashbox{2}(14,14){$P_{10}$}}
\put(193,13){\dashbox{2}(14,14){$P_{11}$}}
\put(207,80){\line(1,0){6}}
\put(207,20){\line(1,0){6}}

\put(213,13){\framebox(14,14){$\C_r$}}
\put(213,33){\framebox(14,14){$\C_r$}}
\put(213,53){\framebox(14,14){$\C_r$}}
\put(213,73){\framebox(14,14){$\C_r$}}

\put(227,20){\line(1,0){6}}
\put(227,40){\line(1,0){6}}
\put(227,60){\line(1,0){6}}
\put(227,80){\line(1,0){6}}

\put(233,13){\dashbox{2}(14,14){$P_{12}$}}
\put(233,33){\dashbox{2}(14,14){$P_{13}$}}
\put(233,53){\dashbox{2}(14,14){$P_{14}$}}
\put(233,73){\dashbox{2}(14,14){$P_{15}$}}

\put(247,20){\line(1,0){6}}
\put(247,40){\line(1,0){6}}
\put(247,60){\line(1,0){6}}
\put(247,80){\line(1,0){6}}

\put(253,13){\framebox(14,14){$H$}}
\put(253,33){\framebox(14,14){$H$}}
\put(253,53){\framebox(14,14){$H$}}
\put(253,73){\framebox(14,14){$H$}}

\put(267,20){\line(1,0){6}}
\put(267,40){\line(1,0){6}}
\put(267,60){\line(1,0){6}}
\put(267,80){\line(1,0){6}}

\put(273,13){\dashbox{2}(14,14){$P_{16}$}}
\put(273,33){\dashbox{2}(14,14){$P_{17}$}}
\put(273,53){\dashbox{2}(14,14){$P_{18}$}}
\put(273,73){\dashbox{2}(14,14){$P_{19}$}}

\put(287,20){\line(1,0){8}}
\put(287,40){\line(1,0){8}}
\put(287,60){\line(1,0){8}}
\put(287,80){\line(1,0){8}}

\put(295,80){\line(3,2){10}}
\put(295,80){\line(3,-2){10}}
\put(295,60){\line(3,2){10}}
\put(295,60){\line(3,-2){10}}
\put(295,40){\line(3,2){10}}
\put(295,40){\line(3,-2){10}}
\put(295,20){\line(3,2){10}}
\put(295,20){\line(3,-2){10}}

\end{picture}
\caption{The protocol for the state $\ket{\psi_{4,Z}}$ using randomized compiling to convert all errors except those in the channel $\C_r$ into stochastic errors.  The Paulis $P_i$ (in dashed boxes to distinguish them from the original gates of the circuit) combine a random Pauli for the preceding gate with a random Pauli for the subsequent gate conjugated by that gate.  For example, if the random Pauli for the preparation of the second qubit is $Z$ and the random Pauli for the first CNOT gate is $X \otimes X$, then $P_3 = Z$ and $P_5 = X$.}
\label{fig:GHZtwirled}
\end{figure}

One small but important difference from regular randomized compiling is that after creating the appropriate state from our maximally sensitive set, the Pauli frame must return to the identity.  That is, in the absence of noise, the circuit should create the actual desired state from the maximally sensitive set, not some Pauli twirled version of it.  In particular, we don't twirl the channel $\C_r$.  Normally randomized compiling would be applied to the full circuit, and any intermediate state would be therefore be altered from its unrandomized version.

In the case of the GHZ states, it would be acceptable to have $\ket{00 \ldots 0} - \ket{11 \ldots 1}$ in whichever basis instead of the $+$ version, but $\ket{x} + \ket{\overline{x}}$ for $x \neq 00 \ldots 0$ would not work.  Instead of seeing an error rate proportional to $n^2$, the rotations on the $1$ bits of $x$ would cancel with the rotations on the $0$ bits of $x$, and the signal would be reduced to something proportional to $(2 \mathrm{wt}(x) - n)^2$.  In the typical case where $\mathrm{wt}(x) = n \pm O(\sqrt{n})$, this means the scaling of errors would just be linear in $n$; we would have eliminated the signal along with the extraneous gate noise.  This is why it is important to have all Paulis cancel before applying the channel.

We twirl all gates used in the circuit for encoding or measurement. We also must twirl state preparation and measurement, but that is a bit more subtle since we can only insert a Pauli after the state preparation but not before, and before the measurement but not after.  For measurement in the standard basis, we could still insert a random Pauli as before and compensate by classically reversing the measurement outcome if the Pauli before the measurement was an $X$ or a $Y$.  Note that a $Z$ just before the measurement would have no effect on the measurement outcome if the measurement were perfect, but could conceivably alter the result when there is noise in the measurement.

Twirling out the errors in the state preparation requires a bit more care.  One option is symmetric to the way we handle measurements: put a random Pauli after the state preparation and compensate by preparing a $\ket{1}$ instead of a $\ket{0}$ (or vice-versa) if the Pauli is an $X$ or a $Y$.  However, for this to work, it requires that preparation of both $\ket{0}$ and $\ket{1}$ can be modeled by an ideal state preparation followed by the same noise channel $\D_s$ for both states.  This is an additional non-trivial assumption.

Note, though, that phase errors after an ideal state preparation of a basis state have no physical effect.  This means that we can ignore them and focus instead on twirling the bit flip errors, which only requires a randomly chosen $I$ or $Z$.  More precisely, suppose that instead of preparing the state $\ket{0}$, the noisy state preparation procedure prepares $\rho$.  If we randomly apply $Z$ $50\%$ of the time to this state, the resulting state is
\begin{equation}
\rho' = \frac{1}{2} \rho + \frac{1}{2} Z \rho Z = \rho_{00} \ketbra{0}{0} + \rho_{11} \ketbra{1}{1}.
\end{equation}
That is, the state is a mixture of the correct $\ket{0}$ state with probability $\rho_{00}$ and the state with a bit flip $\ket{1}$ with probability $\rho_{11}$.  This error has no coherent part and cannot contribute to a quadratic accumulation of error.  In fact, essentially the same argument applies to measurements: We do not need to twirl with $X$ and $Y$, only with $I$ and $Z$.  

The upshot is that provided the noise in the Pauli gates is $o(r)$ and the noise in all other circuit components is $O(r)$ and fits the other assumptions of the error model, we can perform the GHZ maximally sensitive states protocol described in sec.~\ref{sec:maxsensitive} with the addition of Pauli twirling, and if we observe a quadratic scaling of noise rate with $n$, the number of qubits, that indicates that there are coherent errors in the channel.  The coefficients of the scaling for the three bases tell us the parameters of the coherent part of the error, as before.

\section{Testing for Coherent Errors in Gates and Measurements}
\label{sec:testgates}

\subsection{Testing Gates}

The protocol can be extended to test for coherent errors in certain kinds of quantum gates, with potential additional assumptions on the noise.  Suppose we wish to test one type of gate for coherent errors.  We again model the noisy gate as a perfect unitary $F$ followed by one member of a family of quantum channels $\C_r$.  The basic idea is that instead of creating the states $\ket{\psi}$ from the maximally sensitive set, we create the states $(F^\dagger)^{\otimes n} \ket{\psi}$, then apply $F^{\otimes n}$ to get $\ket{\psi}$ and measure the fidelity to $\ket{\psi}$ as before.

Clearly, it only makes sense to do this in a situation where we allow gate errors.  We use the model and procedure in sec.~\ref{sec:gatenoise}.  That is, the assumption is that gate noise is Markovian and specific to the qubits involved in the gate, that the noise in Pauli operations is negligible, and that errors in other gates are $O(r)$, with $r$ the worst-case fidelity of the gate being tested.

When $F$ is a single-qubit Clifford group gate, the twirling procedure from sec.~\ref{sec:gatenoise} works without any further alteration.  Indeed, we can create one of the GHZ states, twirling the initial state preparations and gates involved in the circuit as before, then perform $F^\dagger$ on each qubit, twirling it as well.  Since $F^\dagger$ is also in the Clifford group, this still requires only Pauli operations before and after.  Then we apply $F$ to each qubit \emph{without} twirling.  Since we model noisy gates as the perfect gate followed by a channel, the effect is the same as applying the channel directly on the relevant GHZ state.  Then the measurement and analysis is the same as for testing for coherent noise in a channel.

We get the following protocol:
\begin{prot}[Identifying coherent errors in a single-qubit Clifford group gate $F$]
Pick a range of values of $n$, as in protocol~\ref{prot:channel}.
\begin{enumerate}
\item For each value of $n$, use randomized compiling to create copies of $\ket{\psi_{n,X}}$, $\ket{\psi_{n,Y}}$, and $\ket{\psi_{n,Z}}$.
\item Apply $F^\dagger$ twirled to each qubit.  (I.e., choose a random Pauli $P$ and apply $P F^\dagger Q$, where $Q = FPF^\dagger$.)
\item Apply $F$ to each qubit without twirling it.
\item Measure the qubits of $\ket{\psi_{n,X}}$ and $\ket{\psi_{n,Y}}$ in the $Z$ basis and measure the qubits of $\ket{\psi_{n,Z}}$ in the $X$ basis.
\item Count the outcome as an error if the outcome bit string for a particular state has odd parity.  Determine the error rate for each state.
\item Fit the error rates for different $n$ for each basis $P$ to a quadratic $a_P n^2 + b_P n + c_P$.
\item Let $\theta^2 = a_X + a_Y + a_Z$ and $v_P^2 = a_P / \theta^2$.  These are the parameters of the coherent error in the noise channel for $F$. 
\end{enumerate}
\end{prot}

When $F$ is a non-Clifford group gate, an additional assumption is needed.  To twirl, we need to perform $Q = F P F^\dagger$ for a random Pauli $P$, and this may not be a Pauli operator, and therefore our current noise model allows it to have a significant level of coherent noise.  While $Q$ will not be needed on every qubit, non-Clifford group gates will be needed for at least $1/2$ of the qubits, which is enough to potentially see a quadratic scaling with $n$ of coherent noise.  Since $(F^\dagger)^{\otimes n} \ket{\psi_{n,P}}$ may not be a stabilizer state in this case, it cannot be created using only Clifford group gates.

Therefore, we will make a further assumption on the noise model.  Suppose $F$ is a single-qubit $C_3$ gate such as the $\pi/8$ phase rotation $\mathrm{diag}(e^{-i\pi/8}, e^{i\pi/8})$, which means that $F P F^\dagger$ is in the Clifford group whenever $P$ is in the Pauli group.  For the twirl, we therefore only need to add Clifford group gates to the ideal circuit creating $(F^\dagger)^{\otimes n} \ket{\psi_{n,P}}$.  However, we must be able to perform the new Clifford group gates with negligible noise as well as the Paulis.  For instance, when $F$ is the $\pi/8$ phase rotation, then the twirl sometimes uses the $\pi/4$ phase rotation $F^2$, which is a Clifford group gate.  More generally, we can test the non-Clifford gate $F$ reliably if the gates $F P F^\dagger$ have low noise for all Paulis $P$, with ``low'' meaning $o(r)$.  If we have at least one such gate, then together with the Clifford group, it forms a universal set of gates~\cite{NRS} and we can make the needed states $(F^\dagger)^{\otimes n} \ket{\psi_{n,P}}$ by approximating $F^\dagger$ with our universal set.

Next, we consider two-qubit gates.  To fully test them, we use a maximally sensitive set of states not for qubits but for $4$-dimensional qudits.  To test $F$, we again wish to create the state $(F^\dagger)^{\otimes n} \ket{\psi}$, now treating each $F$ gate as a single-qudit gate by pairing the two qubits into a qudit.  If $F$ is a Clifford group gate for qubits (such as the CNOT gate), it makes the most sense to use the GHZ states $\ket{\psi_{n,B}}$ for $B$ running over some set of MUBs.  Then everything else works the same as for testing a single-qubit gate.

\subsection{Testing Measurements}

Testing measurements is a bit more subtle.  A noisy single-qubit measurement can be written as a $2$-outcome POVM: $\{M_0, M_1\}$.  Each $M_i$ is positive and $M_0 + M_1 = I$.  We can diagonalize $M_0$, which makes $M_1$ diagonal as well.  Let $U$ be the change of basis; then $M_0 = U D U^\dagger$, $M_1 = U (I-D) U^\dagger$.  Here $D$ is diagonal with its two entries in $[0,1]$ to make sure both $M_0$ and $M_1$ are positive.

The $D$ factor is a probability that the measurement outcome will be wrong; it can only ever accumulate linearly.  The $U$ is a change of basis for the measurement and can create quadratic accumulation.  Our goal is to measure that term.  Let us consider what happens if we use the protocol based on the set of three GHZ states.  When we prepare $\ket{\psi_{n,X}}$ or $\ket{\psi_{n,Y}}$, the next thing to do is to measure each qubit in the $Z$ basis.  Thus, the change of basis for the measurement is equivalent to applying $U$ to each individual qubit of the GHZ state.  However, when we prepare $\ket{\psi_{n,Z}}$, we do a Hadamard on each qubit before measuring.  The change of measurement basis is now equivalent to $HUH$ on each qubit of the GHZ state.  When a different channel applies for different states in the maximally sensitive set, the protocol may not work.  In particular, the states $\ket{\psi_{n,X}}$ or $\ket{\psi_{n,Y}}$ measure $v_X$ and $v_Y$ for $U$, but $\ket{\psi_{n,Z}}$ measures $v_Z$ for $HUH$, which is the same as $v_X$ for $U$.  Thus, we have redundant information about $v_X$ and none about $v_Z$.

However, notice that the effect of $v_Z$ is to perform a phase rotation, which will have no effect on the measurement; we could absorb it into $D$ and not change either $D$ or $M_i$ at all.  Thus, we might as well assume that $v_Z = 0$.  Measuring $\ket{\psi_{n,X}}$ and $\ket{\psi_{n,Y}}$ is then sufficient to completely characterize $U$.

To be more precise, here is the protocol: 
\begin{prot}[Identifying coherent errors in a single-qubit measurement]
Pick a range of values of $n$, as in protocol~\ref{prot:channel}.
\begin{enumerate}
\item For each value of $n$, use randomized compiling to create copies of $\ket{\psi_{n,X}}$ and $\ket{\psi_{n,Y}}$.
\item Measure the qubits of each state in the $Z$ basis.
\item Count the outcome as an error if the outcome bit string for a particular state has odd parity.  Determine the error rate for each state.
\item Fit the error rates for different $n$ for each basis $P$ to a quadratic $a_P n^2 + b_P n + c_P$.
\item Let $\theta^2 = a_X + a_Y$ and $v_P^2 = a_P / \theta^2$ for $P = X, Y$.  Set $v_Z = 0$.  These are the parameters of the coherent error in the measurement. 
\end{enumerate}
\end{prot}

\subsection{State Preparation}

State preparation cannot be tested for coherent errors in this approach.  Certainly, once the GHZ state is created, there is no need or scope for further state preparation, but the difficulty is more fundamental: state preparation can never have quadratic accumulation by itself.

To see this, we model the noisy state preparation as preparing a state $\rho$ instead of a state $\ket{0}$.  We can diagonalize $\rho = U \sigma U^\dagger$.  This is equivalent to creating the mixed state $\sigma$, a mixture of $0$ and $1$ that cannot cause quadratic error accumulation, followed by the unitary rotation $U$.  The $U$ is a coherent error, and it is there we would be looking for quadratic accumulation.

The problem is that when we prepare $n$ qubits $\rho^{\otimes n}$, the fidelity to the ideal state $\ket{0 \ldots 0}$ is $(1-r)^n \approx 1-nr$.  Any measurement cannot distinguish the faulty state from the ideal state with probability greater than $nr + o(r)$.  This is just linear accumulation.

This is not to say that state preparation cannot contribute to quadratic accumulation of errors or that coherent errors in state preparation are meaningless.  For instance, if we prepare the state $\rho$ and then wait $n-1$ time steps, during each of which there is also an error $U$, we would have an overall state $U^n \sigma U^n$.  With our standard parametrization of $U$, that leads to an infidelity of roughly $n^2 \theta^2$.  Note that this is $n^2$ and not $(n-1)^2$ --- the $U$ from the state preparation does contribute to the quadratic accumulation in this case.  However, the bulk of the coherent error comes from the channel errors during the waiting, not the state preparation.

\section{Accumulation in More General Circuits}
\label{sec:circuits}

The example of state preparation shows that we need to consider linear versus quadratic accumulation in more general contexts than simply applying $\C_r^{\otimes n}$ to entangled states.  In particular, one might ask if there are channels that do not show quadratic accumulation when applied in parallel but that do show quadratic accumulation when applied sequentially or as part of a larger circuit.  One might also wonder if there are channels that can accumulate errors rapidly in conjunction with other different channels even though they do not do so with multiple copies of themselves.

To answer both of these questions, we need to define accumulation of channels in more general circuits.  In order to have this be sensible, we still need to deal with families of channels parametrized by worst-case fidelity $r$ and work with the limit $r \rightarrow 0$; only then do we take the limit of $n \rightarrow \infty$, where $n$ is now the number of channels (possibly different channels) used in a possibly large circuit.  This means we will need an infinite family of circuits, and in order to talk about the interaction of different channels, we need an infinite family of channels, all of which scale in roughly the same way.

For the purposes of this section only, I will refer to the definition of accumulation from sec.~\ref{sec:quadratic} as \emph{accumulation in parallel}.  I will show that this can occur for the same channels as accumulation in the more general context and therefore there is actually no need to distinguish the cases.

\begin{dfn}
\label{dfn:sequentialaccumulation}
A \emph{nested family of circuits} $\{Q_n\}$ is a sequence of quantum circuits for which the circuit $Q_n$ begins with the initial state $\ket{0 \ldots 0}$ for $m_n$ qubits, then performs the sequence of unitaries $U_0, U_1, U_2, \ldots, U_n$ in order.  In particular, $U_q$ is the same for all circuits $Q_n$ in the family for $n \geq q$.  Each $U_q$ is unitary and may act on any finite (but potentially unbounded) number of qubits.  Let $\ket{\psi_n}$ be the final state produced by the circuit $Q_n$, so $\ket{\psi_0} = U_0 \ket{0 \ldots 0}$ and $\ket{\psi_n} = U_n \ket{\psi_{n-1}}$.

Let $\{\C_{j,r}\}$ be a sequence of quantum channels analytic in $r$, with $r_j(r)$ the worst-case fidelity of $\C_{j,r}$ and $r = \sup_j {r_j(r)}$.  We assume that $r_j = \Theta(r)$ for all $j$.  The channels may act on multiple qubits at a time, but we assume there is some constant ($n$ and $r$-independent) bound on the number of qubits they act on, usually $2$.

Let $\{B_i\}$ be a sequence of linear operations on density matrices.  These might be CPTP maps or not.  They might be presented as linear operators $A$ on Hilbert space, which can also be interpreted as linear operations on density matrices $\rho \mapsto A \rho A^\dagger$; this should be assumed.  Let $q_i$ be non-negative integers with $q_i < q_{i+1}$.  Then $Q_n (B_1, q_1; B_2, q_2; \ldots; B_j, q_j)$ is the noisy circuit created from $Q_n$ by inserting $B_i$ after $U_{q_i}$.  That is, we get this noisy circuit by starting with the state $\ket{\psi_{q_1}}$, then performing $B_1$, then the unitaries $U_{q_1 + 1}, \ldots , U_{q_2}$, then $B_2$, and so on through $B_j$, and finishing off with $U_{q_j + 1}$ through $U_n$.  Let $\rho_n (B_1, q_1; B_2, q_2; \ldots; B_j, q_j)$ be the state produced by $Q_n (B_1, q_1; B_2, q_2; \ldots; B_j, q_j)$.   When all $B_j$ are linear operators on Hilbert space, the state is a pure state $\ket{\phi_n (B_1, q_1; B_2, q_2; \ldots; B_j, q_j)}$.  

Let $\rho_{n,r} = \rho_n (\C_{0,r}, 0; \C_{1,r},1; \ldots; C_{n, r}, n)$.  That is, $\rho_{0,r} = \C_{0,r} (\ketbra{\psi_0}{\psi_0})$, the state after $Q_0$ with channel $\C_{0,r}$, and $\rho_{n,r} = \C_{n,r} (U_n \rho_{n-1,r} U_n^\dagger)$, the state after $Q_n$ with channels $C_{j,r}$.

Let 
\begin{equation}
F (n,r) = \bra{\psi_n} \rho_{n,r} \ket{\psi_n}
\end{equation}
be the fidelity between the state produced by the noisy circuit  $Q_n (\C_{0,r}, 0; \C_{1,r},1; \ldots; \C_{n, r}, n)$ and the ideal circuit $Q_r$.  The dependence on the circuits and channels is suppressed here.  Let $G(n) = \lim_{r \rightarrow 0} (1-F(n,r))/r$ if the limit exists.  Then we say that \emph{the sequence of channels $\{\C_{j,r}\}$ has order $a$ accumulation in the circuit family $\{Q_n\}$} if $G(n) = \Theta(n^a)$.
\end{dfn}

This definition is similar to the definition of accumulation in parallel, but does not quite subsume it in general, since we might have a sequence of states $\ket{\psi_n}$ achieving quadratic accumulation, but which cannot be produced by a nested circuit family.  It does, however, include the GHZ states.  To see this, let us consider the following circuit family: $m_n = n+2$.  $Q_0$ creates the EPR pair $\ket{00} + \ket{11}$, and $U_n$ is a CNOT gate from qubit $n$ to qubit $n+1$.  Thus, $\ket{\psi_n}$ is the $(n+2)$-qubit GHZ state.  Let $\C_{n,r}$ be the channel $\C_r$ acting on qubit $n$.  Then $\rho_{n,r}$ is what we get from $\C_r^{\otimes (n+1)}$ acting on the $(n+2)$-qubit GHZ state, which is what we considered in sec.~\ref{sec:maxsensitive} (up to one copy of $\C_r$, which will not affect the order of accumulation).

While the definition assumes that the circuits $Q_n$ are unitary, it also covers circuits including non-unitary channels and measurements, since we can purify those to make a larger unitary circuit.  The unitaries $U_n$ can act on an unbounded number of qubits partially to accommodate this and partially to allow many gates between insertions of the channels.  Note that this also means the definition can encompass cases where there are other sources of noise present and we wish to define the accumulation of errors from a particular sequence of channels relative to the background.

It may be that a particular nested family of circuits prevents the channels from adding fully coherently; this could result in accumulation somewhere strictly between $1$ and $2$, or it could be even less than $1$ in some cases.  It is also possible that the accumulation will not be well-defined because the limit for $G(n)$ is not well-defined.
However, it will turn out that the accumulation can never be greater than $2$, and indeed, if we take the maximum accumulation achievable over all nested families of circuits, that maximum will always be either $1$ or $2$, as for accumulation in parallel.

Suppose we have a circuit family and sequence of channels as above such that the accumulation is defined.  In this context, it may be that some of the channels contribute to the leading order accumulation and some do not, for instance because they are effectively twirled by the circuit relative to the other channels.  We can analyze this by looking at the fidelity to the states of the perfect circuit family when that particular channel is present versus when it is not:
\begin{dfn}
\label{dfn:contributesaccumulation}
Let $\{Q_n\}$, $\{\C_{j,r}\}$, and $G(n)$ be as in def.~\ref{dfn:sequentialaccumulation}.  Let $\rho_{n,r}^j$ be the state produced by the noisy circuit $Q_n (\C_{0,r}, 0; \C_{1,r},1; \ldots; \C_{j-1,r}, j-1; \C_{j+1,r}, j+1; \ldots; \C_{n, r}, n)$, that is the noisy circuit containing all channels except for $\C_{j,r}$.  Let
\begin{equation}
F^j (n,r) =  \bra{\psi_n} \rho_{n,r}^j \ket{\psi_n}
\end{equation}
and $G^j (n) = \lim_{r \rightarrow 0} (1-F^j(n,r))/r$ if the limit exists.  If both $G(n)$ and $G^j(n)$ are well-defined, let $\Delta_j (n) = G(n) - G^j(n)$.  When $\Delta_j(n) = \Theta(n^b)$, then the channel family $\C_{j,r}$ \emph{contributes order $b$ coherent accumulation with $\{\C_{j,r}\}$ in the circuit family $\{Q_n\}$}.
\end{dfn}
Equivalently, we could look at 
\begin{equation}
\delta F^j (n,r) = \bra{\psi_n} (\rho_{n,r}^j - \rho_{n,r}) \ket{\psi_n}
\end{equation}
and then let $\Delta_j(n) = \lim_{r \rightarrow 0} \delta F^j(n,r)/r$.

Since $G^j (n)$ involves one fewer channel than $G(n)$, in the case where all channels are contributing equally to the accumulation of noise, then $G^j(n)$ should be equal to $G(n-1)$.  When $G(n) = Cn^a$, that means that $\Delta_j (n)$ would be $O(n^{a-1})$.  It might be, however, that this particular channel contributes more or less noise accumulation than the other channels, and this definition helps to identify that.  It could also be that this particular channel \emph{decreases} the noise, for instance via a coherent rotation in the opposite direction from the other channels.  Then $\Delta_j(n)$ will be negative.  In the case where the overall accumulation is quadratic, the intution is that most pairs of channels will constructively interfere to produce a rapid divergence from the noiseless state.  The channel $\C_{j,r}$ contributes to this if it constructively interferes with a constant fraction of the other channels.  Conversely, if it only constructively interferes with a constant number of other channels, or with none at all, then it only contributes a constant amount of coherent accumulation.

\begin{thm}
\label{thm:sequential}
If a channel family $\C_{j,r}$ contributes order $b$ coherent accumulation with $\{\C_{j,r}\}$ in the circuit family $\{Q_n\}$ with $b > 0$, then $b \leq 1$ and $\C_{j,r}$ has quadratic accumulation in parallel, as does $\C_{q,r}$  for $\Omega(n^b)$ values of $q \neq j$.  If the maximum value of $b$ for all channel families in $\{\C_{j,r}\}$ is $b_{\mathrm{max}}$, then the sequence of channels has at most order $b_{\mathrm{max}}+1$ accumulation in the circuit family $\{Q_n\}$.
\end{thm}

If a channel has quadratic accumulation in parallel, it can also contribute order 1 coherent accumulation in sequence (as the example of GHZ states shows), but that does not mean it always does so in any particular family of circuits.  It may be that in a particular family, there are no other channels that it can constructively interfere with, or it may get twirled out relative to the other channels, or perhaps there is both constructive and destructive interference at work which largely cancels out.

\begin{proof}
Once again, we can make a leading Kraus approximation and consider only $A_0$ for each channel:  Let $A_k^q$ be the $k$th Kraus operator for channel $q$.  Let $\rho^{q}_{k} = \rho_n (A_0^0, 0; A_0^1, 1; \ldots; A_0^{q-1},q-1; A_k^q; q; A_0^{q+1}, q+1; \ldots; A_0^n, n)$, the noisy circuit where we have made the leading Kraus approximation in all but the location $q$, where we have instead used the Kraus operator $A_k^q$.  Let $\rho_0 = \rho_n (A_0^0, 0; A_0^1, 1; \ldots; A_0^n, n)$, where we make the leading Kraus approximation in all locations.  Let $\rho^{q,j}_k$ ($q \neq j$) and $\rho_0^j$ be these same quantities but where we omit inserting any error in the $j$th location.

Since $\|A_k^j\|_\infty = O(\sqrt{r})$ for $k > 0$, then $\| \rho^j_k \| = O(r)$; if we were to look at states of the circuits where two locations did not use the leading Kraus approximation, those would be $O(r^2)$.  We then have
\begin{align}
\rho_{n,r} &= \rho_0 + \sum_{q=0}^n \sum_{k \neq 0} \rho^q_k + O(r^2) \\
\rho_{n,r}^j &= \rho_0^j + \sum_{q\neq j} \sum_{k \neq 0} \rho^{q,j}_k + O(r^2) \\
\rho_{n,r}^j - \rho_{n,r} &= (\rho_0^j -  \rho_0) - \sum_{k \neq 0} \rho^j_k + \sum_{q \neq j} \sum_{k \neq 0} (\rho^{q,j}_k -  \rho^q_k) + O(r^2).
\end{align}
Now, $A_0^q$ is close to $I$ for all $q$, $A_0^q = I + o(1)$.  Since $\rho^q_k$ and $\rho^{q,j}_k$ are $O(r)$, it follows that for $k \neq 0$,
\begin{align}
\rho^q_k &= \rho_n (A_k^q; q) + o(r) \\
\rho^{q,j}_k &= \rho_n (A_k^q; q) + o(r),
\end{align}
so $\rho^{q,j}_k -  \rho^q_k = o(r)$.  Thus,
\begin{equation}
\rho_{n,r}^j - \rho_{n,r} = (\rho_0^j -  \rho_0) - \sum_{k \neq 0} \rho^j_k + o(r).
\end{equation}
The term $\sum \rho^j_k = \sum \rho_n (A_k^q; q) + o(r)$ is just a constant size, independent of $n$, so its contribution to $\Delta_j(n)$ will be a constant.  If the channel $\C_{j,r}$ is contributing non-constant accumulation, it must come from the term $\rho_0^j -  \rho_0$.

Since the worst-case fidelity of $\C_{j,r}$ is $\Theta(r)$, by equation~\eqref{eq:A0worstcasequdit}, it is still the case that at least one of $p$ and $\theta^2$ for each channel must be $\Theta(r)$.  Expand each $A_0^q$ as in equation~\eqref{eq:A0qudit}:
\begin{equation}
A_0^q =  (1 - p_q - \theta_q^2/2) I + \sum_{P \neq I} (i \theta_q h_{q,P} + d_{q,P}) P.
\end{equation}
Since all of the $A_0^q$ are linear operators on Hilbert space, the states $\rho_0 = \ketbra{\phi_0}{\phi_0}$ and $\rho_0^j = \ketbra{\phi_{0}^j}{\phi_0^j}$ are pure states.  We find (to order $r$)
\begin{align}
\ket{\phi_0} &= \left[1 - \sum_{q=0}^n (p_q + \theta_q^2/2)\right] \ket{\psi_n} + \sum_{q=0}^n \sum_{P \neq I} (i \theta_q h_{q,P} + d_{q,P}) \ket{\phi_n (P, q)} - \nonumber \\
& \qquad - \sum_{q < q'} \sum_{P, Q \neq I} \theta_q \theta_{q'} h_{q,P} h_{q',Q} \ket{\phi_n (P,q; Q, q')}.
\end{align}
$\ket{\phi_0^j}$ is similar, but with the sums taken only over $q \neq j$.
Then
\begin{align}
\bra{\psi_n} \rho_{0} \ket{\psi_n} &= |\braket{\psi_n}{\phi_0}|^2 \\
&= 1 - \sum_{q=0}^n (2p_q + \theta_q^2) + \sum_{q=0}^n \sum_{P \neq I}  [(i \theta_q h_{q,P} + d_{q,P}) \braket{\psi_n}{\phi_n (P, q)} + \mathrm{c.c.}] + \\
&  + \Big| \sum_{q=0}^n \sum_{P \neq I} \theta_q h_{q,P} \braket{\psi_n}{\phi_n (P, q)} \Big|^2 - 2\sum_{q<q'} \sum_{P,Q \neq I}  \theta_q \theta_{q'} \text{Re } \big(h_{q,P} h_{q',Q} \braket{\psi_n}{\phi_n (P,q; Q, q')}\big). \nonumber 
\end{align}

Now, the circuits $Q_n$ and $Q_n (P,q)$ both have the same sequence of unitaries $U_{q+1}, \ldots, U_n$ after location $q$, and $\ket{\phi_k (P,q)} = P \ket{\psi_q}$, so
\begin{equation}
\braket{\psi_n}{\phi_n (P, q)} = \bra{\psi_q} P \ket{\psi_q}
\end{equation}
and $\braket{\phi_n (P, q)}{\psi_n} = \bra{\psi_q} P^\dagger \ket{\psi_q}$.  By reordering the sums, we therefore get
\begin{equation}
\sum_{P \neq I}  h_{q,P} \braket{\psi_n}{\phi_n (P, q)} = \sum_{P \neq I}  h_{q,P^\dagger}  \braket{\phi_n (P, q)}{\psi_n},
\end{equation}
and recall that $h_{q,P^\dagger} = h_{q,P}^*$.
Thus,
\begin{align}
\bra{\psi_n} \rho_{0} \ket{\psi_n}  &= 1 - \sum_{q=0}^n (2p_q + \theta_q^2) + \sum_{q=0}^n \sum_{P \neq I} 2  \text{Re } \big(d_{q,P} \bra{\psi_q}P \ket{\psi_q} \big) + \\
& + \Big| \sum_{q=0}^n \sum_{P \neq I} \theta_q h_{q,P}\bra{\psi_q}P \ket{\psi_q} \Big|^2 - 2\sum_{q<q'} \sum_{P,Q \neq I}  \theta_q \theta_{q'} \text{Re } \big(h_{q,P} h_{q',Q} \braket{\psi_n}{\phi_n (P,q; Q, q')} \big). \nonumber 
\end{align}
We find the same expression for $\bra{\psi_n} \rho_0^j \ket{\psi_n}$, but the sums for $q$ and $q'$ exclude $q=j$ or $q'=j$.  Therefore,
\begin{align}
\delta F^j (n,r) &= \bra{\psi_n} (\rho_0^j - \rho_0) \ket{\psi_n} - \sum_{k \neq 0} \bra{\psi_n} \rho_k^j \ket{\psi_n} + o(r) \\
&= A^j (r) +  B^j (n,r) + o(r) \label{eq:deltaF} \\
A^j (r) &=  2p_j + \theta_j^2 - \sum_{P \neq I} \left[ 2  \text{Re } \big(d_{j,P}\bra{\psi_j}P \ket{\psi_j} \big) - \theta_j^2 |h_{j,P}|^2 |\bra{\psi_j}P \ket{\psi_j}|^2 \right] -  \sum_{k \neq 0} |\bra{\psi_j} A_k \ket{\psi_j}|^2 \\
B^j (n,r) &=  2 \sum_{q \neq j}\sum_{P,Q \neq I} \theta_q \theta_j \text{Re } \Big( h_{q,P} h_{j,Q} \braket{\psi_n}{\phi'_n (P,q; Q, j)} - h_{q,P} h_{j,Q}^* \bra{\psi_q}P \ket{\psi_q} \bra{\psi_j}Q^\dagger \ket{\psi_j} \Big)
\end{align}
In this sum, $\ket{\phi'_n (P,q; Q, j)} = \ket{\phi_n (P,q; Q,j)}$ if $q < j$ and $\ket{\phi'_n (P,q; Q, j)} = \ket{\phi_n (Q,j; P,q)}$ if $j < q$.

All terms in $A^j(r)$ are $O(r)$ and only involve sums over constant (in $n$) sized domains, so only $B^j(n,r)$ can contribute order $b$ coherent accumulation for $b>0$.  There are at most $O(n)$ terms in $B^j(n,r)$, so $b \leq 1$.  Those terms scale with $r$ only through the product $\theta_q \theta_j$.  Since $\theta_q = O(\sqrt{r})$, the second line will only be linear in $r$ if $\theta_j = \Theta(\sqrt{r})$.  Thus, the channel $\C_{j,r}$ can have $b>0$ only if $\theta_j^2 = \Theta(r)$.  These are exactly the cases in which $\C_{j,r}$ allows quadratic accumulation in parallel.  Each term in the sum for $B^j(n,r)$ can contribute only if $\theta_q = \Theta(\sqrt{r})$ as well, so we additionally need $\theta_q^2 = \Theta(r)$ for $\Omega(n^b)$ values of $q \neq j$.

Also notice that 
\begin{equation}
F(n,r) = 1 - \sum_j A^j (r) - \frac{1}{2} \sum_j B^j (n,r).
\end{equation}
If $b_{\mathrm{max}}$ is the maximum value of $b$ for any channel in the set, then $|B^j(n,r)| = O(n^{b_{\mathrm{max}}})$, which implies $G(n) = O(n^{b_{\mathrm{max}}+1})$.

\end{proof}

\section{Measuring Contribution to Accumulation}
\label{sec:contribute}

Theorem~\ref{thm:sequential} shows that if we wish to identify if a channel can contribute non-trivial coherent accumulation in some family of circuits, it is sufficient to check if it does so in parallel.  However, there are circumstances in which this is not possible.  State preparation illustrates one such case: The channel only appears in conjunction with a particular element (state preparation) which cannot be placed where we need it to be to run the usual maximally sensitive states protocol.  There are other cases of interest too, for instance if we wish to study the noise on a single qubit but are unable to do repeated measurement as in sec.~\ref{sec:fewerqubits}.  Or we might be interested in determining not just whether a particular channel \emph{can} exhibit contribute coherent accumulation, but if it contributes coherently in a specific circuit family.

The basis of our tests will be definition~\ref{dfn:contributesaccumulation} and equation~\eqref{eq:deltaF}.  The basic idea is that we can use randomized compiling to selectively twirl different channels, separately twirling the particular channel to study from other possible noise channels that it can accumulate with.  By comparing the four possibilities of which channels are twirled, we can determine whether the channel of interest is contributing to accumulation.

\subsection{Testing State Preparation}

If we wish to measure the potential for noise in state preparation to contribute coherent accumulation, we want a situation in which the coherent noise in the state preparation step combines with coherent noise elsewhere in the circuit.  One difficulty is that we don't necessarily know which channels are present in the natural noise in the system and even if we did, we don't know under what circumstances they will be able to interfere constructively with the channel we are trying to test.  Therefore, in order to guarantee we can detect all possible coherent noise in the state preparation, we need to add noise of a particular type deliberately.

We need a nested circuit family.  Let us keep the circuits used as simple as possible.  We will simply prepare a $\ket{0}$ state, wait for $n$ time steps, and then measure in the standard basis, but sometimes we will also add artificial noise, a rotation $\exp(i \phi \vec{w} \cdot \vec{\sigma})$, after each time step.  I will assume for the time being that the naturally occurring noise during the waiting time is negligible, and that all the naturally occurring noise is in the state preparation and measurement.   In particular, I am assuming that there is no noise in the rotations, but the result is not at all sensitive to this assumption.  We wish to check for coherent errors in the preparation, but not in the measurement, so we will always twirl the measurement: With probability $1/2$ we do the circuit as above, and with probability $1/2$ do the circuit but with a $Z$ just before the measurement.  The noisy twirled measurement, as discussed above, can be modeled as a channel $\D$ followed by a perfect measurement.

Let $Q_n$ be the family of circuits consisting of perfect $\ket{0}$ state preparation followed by $n$ time steps of waiting.  $\{Q_n\}$ is then a nested circuit family.  Let $\C_{0,r}$ be the noise on state preparation.  Then the circuit without artificial noise is $Q_n (\C_{0,r}, 0; \D, n)$ followed by perfect measurement, which is projection on the state $\ket{0}$ and its complement.  This circuit thus exactly measures the fidelity of $\rho_n (\C_{0,r}, 0; \D, n)$ to $\ket{\psi_0} = \ket{0}$.  Let $\C_{q,r}$ be the unitary $\exp(i \phi \vec{w} \cdot \vec{\sigma})$ for $q= 1, \ldots n-1$ and $\C_{n,r}$ be $\exp(i \phi \vec{w} \cdot \vec{\sigma})$ followed by $\D$.  

Let us apply equation~\eqref{eq:deltaF} to $\{Q_n\}$ and $\{\C_{q,r}\}$.  The framework of sec.~\ref{sec:circuits} does not perfectly fit here, since we don't have a fixed sequence of channels for all $n$ --- the $n$th channel is treated differently since it has the noise $\D$ added.  However, this does not matter --- theorem~\ref{thm:sequential} still tells us the conditions under which $\C_{0,r}$ can contribute linear coherent accumulation and equation~\eqref{eq:deltaF} still applies to the particular circuit and channels we have.

The location of interest is $j=0$, so let $\theta_j = \theta$ be the coherent rotation angle for the noise $\C_{0,r}$ during state preparation, and $\theta_q = \phi$ for $j \neq 0$.  Since $\D$ is twirled, this is true also for $q=n$.  Since we are working with a qubit, the $h_{q,P}$ are real, and because $Z$ errors during preparation of $\ket{0}$ are meaningless, $h_{0,Z} = 0$.  For $q \neq 0$, $h_{q,P} = w_P$.  Since $\ket{\psi_0} = \ket{0}$, $\bra{\psi_0} P \ket{\psi_0} = 0$ for $P = X, Y$, the only cases that matter.  Therefore, all the potentially non-constant contribution to $\delta F^0 (n,r)$ in this case will come from the term involving $\braket{\psi_n}{\phi'_n (P,q; Q, j)}$.

Now, $\ket{\phi_n (Q,0; P,q)} =  PQ \ket{0}$.  If $P = Q$, then $PQ = I$.  Otherwise, $PQ = \pm i R$, where $R$ is the third Pauli.  When $R = X, Y$, then $\bra{0} PQ \ket{0} = 0$, so we need only consider the cases when $P = Q$ or $R = Z$.  That is, we have $P = Q = X$, $P = Q = Y$, $P = X$ and $Q = Y$, and $P = Y$ and $Q = X$.  The cases where $R = Z$ give something imaginary, so won't contribute to the real part.  

The $A^j(r)$ term in equation~\eqref{eq:deltaF} is independent of $n$.  We thus have
\begin{align}
\delta F^0 (n,r) &= n \theta \phi \sum_{P,Q \neq I} w_P h_{0,Q} \text{Re } \braket{0}{\phi_n (Q, 0; P, q)} + A^j (r) + o(r) \\
&= n \theta \phi  (w_X h_{0,X} + w_Y h_{0,Y}) + \delta F^0 (0,r)  + o(r).
\end{align}
Therefore if we run some circuits with $w_X = 1$ and some circuits with $w_Y = 1$, we can pick out the $h_{0,X}$ and $h_{0,Y}$ components of the noise on the state preparation.

At least, we could if we could implement both $Q_n^{\phi} = Q_n (\C_{0,r},0; \C_{1,r}, 1; \ldots; \C_{n,r},n)$ and $Q_n^{0,\phi} =  Q_n (\C_{1,r}, 1; \ldots; \C_{n,r},n)$ followed by perfect measurement.  We can do $Q_n^{\phi}$: this is what we get by creating the $\ket{0}$ state using noisy state preparation and following by $n$ time steps of rotations, then a twirled noisy measurement.  However, we cannot do $Q_n^{0,\phi}$, since that would require a perfect (noise-free) state preparation.

What we can do instead is $Q_n^{t,\phi} = Q_n (\C^t, 0; \C_{1,r}, 1; \ldots; \C_{n,r},n)$, where $\C^t$ is the noise channel we get by twirling the state preparation (with probability $1/2$ after preparing the state, do $Z$; otherwise nothing).  Let $\rho^t_{n,r}$ be the output of $Q_n^{t,\phi}$ and let $\delta F^t (n,r) = \bra{0} (\rho^0_{n,r} - \rho^t_{n,r}) \ket{0}$ and $\Delta_t (n) = \lim_{r \rightarrow 0} \delta F^t(n,r)/r$.  $\C^t$ has no coherent part and cannot contribute coherent accumulation, so $\Delta_t(n)$ is guaranteed to be at most constant in $n$.  

Then for $Q_n^{0,\phi}$,
\begin{align}
\delta F^0(n,r) &= \bra{0} (\rho^0_{n,r} - \rho^t_{n,r}) + (\rho^t_{n,r} - \rho_{n,r}) \ket{0} \\
&= \delta F^t (n,r) + \bra{0}  (\rho^t_{n,r} - \rho_{n,r}) \ket{0}.
\end{align}
Since $\bra{0} \rho^t_{n,r} \ket{0}$ is the probability $P_n^t$ that $Q_n^{t,\phi}$ outputs $0$ and $\bra{0} \rho_{n,r} \ket{0}$ is the probability $P_n$ that $Q_n^{\phi}$ outputs $0$, the last terms are measurable.  The remaining term $\delta F^t (n,r)$ is not directly measurable, but the $\Theta(r)$ part of that must be constant in $n$.  In particular, if $P_n^t - P_n$ increases linearly with $n$, the cause must be (subject to our approximations) that the state preparation contributes to linear accumulation.

In order to eliminate the $\delta F^t (n,r)$ term and determine whether the linear accumulation is actually present or not, we can also measure $P_0^t - P_0$.  In particular,
\begin{align}
(P_n^t - P_n) - (P_0^t - P_0) &= [\delta F^0 (n,r) - \delta F^0 (0,r)] + [\delta F^t (0,r) - \delta F^t (n,r)] + o(r) \\
&=  n \theta \phi  (w_X h_{0,X} + w_Y h_{0,Y}) + o(r).
\end{align}
We can then find the parameters of the noise on the state preparation by simply dividing by $n \phi$.  If we want additional confidence that the behavior is actually proportional to $n$, we could measure additional values of $n$.

Putting this all together, we get the following protocol:
\begin{prot}[Measure coherent errors in state preparation]
Choose $n$ and $\phi$ as desired, provided $n\phi \ll 1$ so that we are still in the limit of small $r$.   Subject to this constraint, larger $n\phi$ will make a larger signal $P_{n,X}^t$ and $P_{n,Y}^t$.
\begin{enumerate}
\item Perform $Q_0^{t,\phi}$ (prepare the state $\ket{0}$ and measure immediately, twirling both preparation and measurement), which has output $0$ with probability $P_0^t$.
\item Perform $Q_0^\phi$ (prepare and measure immediately, twirling only measurement), which has output $0$ with probability $P_0$.
\item Perform $Q_n^{t,\phi}$ using $w_X = 1$ and $w_Y = 0$, which has output $0$ with probability $P_{n,X}^t$ and again using $w_X = 0$, $w_Y = 1$, which has output $0$ with probability $P_{n,Y}^t$.
\item Perform $Q_n^\phi$ using  $w_X = 1$ and $w_Y = 0$, which has output $0$ with probability $P_{n,X}$ and again using $w_X = 0$, $w_Y = 1$, which has output $0$ with probability $P_{n,Y}$.
\item Calculate 
\begin{align}
\theta h_{0,X} &= [(P_{n,X}^t - P_{n,X}) - (P_0^t - P_0)]/(n\phi), \\
\theta h_{0,Y} &=  [(P_{n,Y}^t - P_{n,Y}) - (P_0^t - P_0)]/(n\phi).  
\end{align}
$|h_{0,X}|^2 + |h_{0,Y}|^2 = 1$, so this tells us $\theta$, $h_{0,X}$ and $h_{0,Y}$.
\end{enumerate}

\end{prot}

Note that if the $\phi$ angle rotations are slightly noisy, then we have a slightly different circuit than expected: $Q_n (\C_{0,r},0; \C'_{1,r}, 1; \ldots; \C'_{n,r},n)$ with $\C'_{q,r}$ a channel close to the desired $\C_{q,r}$.  The analysis above still holds almost unchanged, but $\phi$, $w_X$, and $w_Y$ might be slightly off from the desired values.  This translates into an error in $h_{0,X}$ and $h_{0,Y}$, but provided the error in the rotation is small, the deduced values will be close to correct.  In principle, one can perform a single rotation by angle $n\phi$ instead of $n$ rotations by an angle $\phi$, but in the case where these rotations are noisy, this might result in a deviation from the desired behavior if the noise cannot be evenly divided up among $n$ separate smaller rotations.  This effect will be particularly pronounced if we are trying different values of $n$ to compare and if rotations by the different angles also have substantially different errors; if the noise on the $n\phi$ rotation does not scale linearly with $n$, this may largely or completely obscure the signal we are trying to see.

Essentially the same procedure can be used with measurement to provide an alternate way of testing for quadratic accumulation of errors during measurement, and a related but slightly more complicated protocol can test for quadratic accumulation of errors in gates.  The protocol for measurement is almost identical to that for state preparation but always twirls the state preparation and either twirls or does not twirl the measurement.  The protocol for gates uses an initial Bell state with an ancilla qubit to ensure that all of $X$, $Y$, and $Z$ can be tested and that only errors that cancel with later errors from the artificial noise will contribute to $B^j(n,r)$ in equation~\eqref{eq:deltaF}.  Compared to the technique from sec.~\ref{sec:testgates}, these protocols are simpler, using just a single qubit (or $2t$ for testing $t$-qubit gates) with no need for reusing qubits or making multiple measurements in the course of the circuit.  However, since the error to be measured only appears once in the circuits and we are subtracting two very similar small quantities, we will need many more iterations of the testing circuits from this section in order to accumulate sufficient statistics to determine if the gate can contribute coherent accumulation or not.

\subsection{Accumulation in a Specific Circuit Family}

One nice thing about the protocol for testing noise in state preparation is that it is not really dependent on seeing a scaling with $n$.  We simply compare the $n=0$ case to a single $n \neq 0$ case and see if there is a difference.  Inspired by this, we can use a similar protocol to determine if a specific location (which could be a gate, a state preparation, or a measurement) in a specific circuit has errors which add coherently with other errors in that circuit.  Since we have no family of circuits and cannot scale to larger $n$, definition~\ref{dfn:contributesaccumulation} does not actually apply in this case, but conceptually, we expect things to behave similarly.

In particular, consider the following protocol:
\begin{prot}[Identify if errors in a particular location add coherently in a circuit]
\ 

\begin{enumerate}
\item Run the circuit $Q^{j,T}$ with randomized compiling for both location $j$ and the set of all locations $T$ that might be experiencing coherent noise.  Note the output distribution $X^{j,T}$.
\item Run the circuit $Q^{T}$ where the location $j$ is not twirled but all other possible coherent noise locations are.  The output distribution is $X^T$.
\item Compute the fidelity $F^{j,T}$ between $X^T$ and $X^{j,T}$; this sets a baseline of comparison for how much difference twirling makes on the noise in location $j$ in isolation.
\item Run the circuit $Q^j$ using randomized compiling on location $j$ but nowhere else.  The output distribution is $X^j$.  Compute the fidelity $F^j$ between $X^j$ and $X^T$.
\item Run the circuit $Q$ not using randomized compiling anywhere.  The output distribution is $X$.  Compute the fidelity $F$ between $X$ and $X^T$.
\item Let $\Delta F = F^j - F$.  If $\Delta F \gg 1-F^{j,T}$, then conclude that the location $j$ is contributing to greater than linear accumulation.
\end{enumerate}
\end{prot}
The idea here is that twirling the location $j$ may change how it contributes to the overall error in the circuit, but not by very much unless it can constructively interfere with noise coming from other locations.  If location $j$ is not coherent with the noise from other locations, the error it creates should be more or less independent of the other errors in the circuit and $\Delta F$ will be about the same as $1-F^{j,T}$.  On the other hand, if there is greater than linear accumulation, that will be present in $Q$ and the overall error rate will be substantially higher than in $Q^j$.

We can further refine this protocol and get some information about what the actual accumulation is (i.e., whether it is quadratic or something smaller but still superlinear) by twirling not all of $T$ but subsets of different sizes.  Then we can look at the scaling of $\Delta F$ as the untwirled subset of $T$ gets larger.

\section{Conclusion}

I have presented an efficient protocol to test for coherent errors in a system.  For coherent noise with worst-case fidelity $r$, we can detect the presence of the noise using only $O(1/\sqrt{r})$ qubits  and operations.  For qubits, the protocol allows us to completely characterize the coherent part of the error in the channel with a similar number of qubits, whereas for higher-dimensional qudits, we can only learn partial information about the coherent part of the channel.  These protocols are within a constant factor (dependent only on the qudit dimension) of the optimal sensitivity to coherent noise, and can thus identify very small errors with considerably fewer resources than protocols which do not explicitly utilize the quadratic scaling of coherent errors.  In general, the protocol is quite efficient, using $O(n)$ gates with small constant factors and straightforward classical processing.

The protocol requires either construction of an $n$-qubit entangled state for modest size $n$ or two qubits ($4$ if measuring the noise in a two-qubit gate) and $n-1$ repetitive measurements and resets of one of those qubits during the circuit.  How large an $n$ is needed depends on the size of the errors and the amount of statistics one is willing to collect, but a modest size $n$, perhaps on the order of $10$, should be enough to distinguish quadratic accumulation of errors from linear accumulation.  I have also presented an alternative protocol that requires even fewer capabilities and can detect coherent errors in state preparation, but is less sensitive, requiring more statistics to see an effect.

Both protocols are robust to small variations in the coherent errors from gate to gate, but when gate errors vary substantially when the same gate is performed on different qubits, the $n$-qubit version of the maximally sensitive states protocol breaks down.  However, the version resetting and reusing one of the qubits (sec.~\ref{sec:fewerqubits}) continues to work without change to detect coherent noise on the gate on the qubit being reset.  The protocol of sec.~\ref{sec:contribute} also works unchanged in this situation.

I have worked in the limit $r \rightarrow 0$ in order to have a well-defined theoretical framework for talking about the dominant error.  What does this mean in practice, for some fixed noise model?  Generally, as we vary $n$, we can expect to have a term linear in $n$ and a term quadratic in $n$, with the latter being the signal we are interested in here and everything else contributing to the linear term.  In order to see the signal, we thus need the coefficient of the quadratic term to be at least about $1/n$ as large as the coefficient of the linear term.  So for instance, if we are working with roughly 10 qubits, we can hope to see the quadratic term even if the effect of coherent errors is roughly $1/10$ the size of the stochastic errors on a single qubit.  This is the advantage of working with maximally sensitive states: the small coherent terms are enhanced in size.  The limits of this are corrections of order $r^2$ compared to the signal of order $r$, plus the scaling will break down once $n^2 r$ becomes of order $1$, so that sets an upper limit to how large an $n$ we can use.   In this same case of approximately 10 qubits, we would expect the protocol to work as long as $r \lesssim 0.01$.

There remain a number of open questions, particularly relating to qudits with dimension $d > 2$.  One question is how to completely characterize the coherent part of a channel for qudits.  Another is how many states are needed.  The most efficient maximally sensitive sets I present take advantage of MUBs, and use $d+1$ states.  It seems likely this is the best possible (though that remains an open question as well), but complete sets of MUBs are only known to exist when $d$ is a prime power.  It is unclear if more states are needed in other dimensions, or if other constructions can achieve similar scaling.

Another open question is the characterization of measurements for qudits.  Noisy qubit measurements are two-outcome POVMs, which have a simple form, but multiple-outcome POVMs are significantly more complicated and may need ancillas to implement.

Perhaps a more critical question is what happens to this protocol when the assumptions about the noise are not satisfied.  I have shown that small variation in the coherent error from qubit to qubit is not too important, but what if there are non-Markovian errors in the other gates?  What if the errors in the Paulis used for twirling are not negligible?  It is plausible that some signal will remain in these cases, but it may well be partially masked by other sources of coherent errors that overcome the randomized compiling.  Leakage errors are another interesting case.  If each qubit is leaking into a separate environment, this is effectively a stochastic error and errors on different qubits cannot accumulate coherently (although they could due to multiple gates on the same qubit).  However, coherent leakage of information into a shared environment might allow quadratic accumulation of error that does not show up in this protocol.

\paragraph*{Acknowledgements}  I would like to thank Arnaud Carignan-Dugas, Joseph Emerson, David Poulin, and Joel Wallman for extremely useful discussions.  Research at Perimeter Institute is supported in part by the Government of Canada through the Department of Innovation, Science and Economic Development Canada and by the Province of Ontario through the Ministry of Economic Development, Job Creation and Trade.

\end{document}